\documentclass{article} 
\usepackage[letterpaper, left=1in, right=1in, top=1in, bottom=1in]{geometry}
\usepackage{natbib}
\bibpunct{(}{)}{;}{a}{,}{,}

\usepackage{hyperref}
\usepackage{url}

\usepackage{paralist}

\newcommand{\autakia}[1]{``#1''}
\newcommand{\step}[1]{^{(#1)} }

\newcommand{\projoperX}[1]{{\Pi_{\calX_{#1}}}}
\newcommand{\projoperY}[1]{{\Pi_{\calY_{#1}}}}
\newcommand{\projoperZ}[1]{{\Pi_{\calZ_{#1}}}}

\usepackage{nicematrix}
\usepackage{amsmath,amsfonts,bm}
\usepackage{multirow,array}
\usepackage{xspace}

\def\eqref#1{equation~\ref{#1}}

\def\1{\bm{1}}

\def\vmu{{\bm{\mu}}}
\def\vtheta{{\bm{\theta}}}

\def\vx{{\bm{x}}}
\def\vy{{\bm{y}}}

\def\mI{{\mathbf{I}}}

\DeclareMathAlphabet{\mathsfit}{\encodingdefault}{\sfdefault}{m}{sl}
\SetMathAlphabet{\mathsfit}{bold}{\encodingdefault}{\sfdefault}{bx}{n}

\newcommand{\R}{\mathbb{R}}

\renewcommand{\vec}[1]{\bm{#1}}
\newcommand{\mat}[1]{\mathbf{#1}}

\newcommand{\calA}{\ensuremath{\mathcal{A}}}
\newcommand{\calB}{\ensuremath{\mathcal{B}}}

\newcommand{\calM}{\ensuremath{\mathcal{M}}}
\newcommand{\calN}{\ensuremath{\mathcal{N}}}
\newcommand{\calO}{\ensuremath{\mathcal{O}}}

\newcommand{\calX}{\ensuremath{\mathcal{X}}}
\newcommand{\calY}{\ensuremath{\mathcal{Y}}}
\newcommand{\calZ}{\ensuremath{\mathcal{Z}}}

\newcommand{\class}[1]{\ensuremath{\mathsf{#1}}\xspace}

\newcommand{\NP}{\class{NP}}

\newcommand{\PLS}{\class{PLS}}
\newcommand{\PPAD}{\class{PPAD}}

\newcommand{\CLS}{\class{CLS}}

\newcommand{\defeq}{:=}

\makeatletter
\newcommand{\xRightarrow}[2][]{\ext@arrow 0359\Rightarrowfill@{#1}{#2}}
\makeatother
\newcommand{\jgda}{\mathbf{J}_{\mathrm{GDA}}}
\newcommand{\jeg}{\mathbf{J}_{\mathrm{EG}}}
\newcommand{\jomwu}{\mathbf{J}_{\mathrm{OMWU}}}

\usepackage{enumerate}
\usepackage{subcaption}
\DeclareRobustCommand{\inlinelist}[1]{\begin{inparaenum}[(i)] #1 \end{inparaenum}}

\usepackage{hyperref}
\newcommand{\declarecolor}[2]{\definecolor{#1}{RGB}{#2}\expandafter\newcommand\csname #1\endcsname[1]{\textcolor{#1}{##1}}}
\declarecolor{Navy}{0, 0, 128}

\hypersetup{
colorlinks=true,
linktocpage=true,
pdfstartview=FitH,
breaklinks=true,
pdfpagemode=UseNone,
pageanchor=true,
pdfpagemode=UseOutlines,
plainpages=false,
bookmarksnumbered,
bookmarksopen=false,
bookmarksopenlevel=1,
hypertexnames=true,
pdfhighlight=/O,
urlcolor=red,linkcolor=Navy,citecolor=Navy,	
pdftitle={},
pdfauthor={},
pdfsubject={},
pdfkeywords={},
pdfcreator={pdfLaTeX},
pdfproducer={LaTeX with hyperref}
}
\usepackage{amsthm}
\theoremstyle{definition}  %

\newtheorem{claim}{Claim}
\newtheorem{lemma}{Lemma}

\newtheorem{proposition}{Proposition}
\newtheorem{fact}{Fact}

\theoremstyle{plain}
\newtheorem{remark}{Remark}

\newtheorem{theorem}{Theorem}
\newtheorem{definition}{Definition}

\numberwithin{claim}{section}
\numberwithin{fact}{section}
\numberwithin{lemma}{section}
\numberwithin{proposition}{section}
\numberwithin{theorem}{section}
\numberwithin{corollary}{section}
\numberwithin{definition}{section}

\usepackage{thm-restate}
\usepackage{stmaryrd}

\usepackage[capitalize,noabbrev]{cleveref}
\crefname{align}{}{}
\crefname{gather}{}{}
\crefname{equation}{}{}
\crefname{claim}{claim}{claims}
\Crefname{claim}{Claim}{Claims}
\crefname{lemma}{lemma}{lemmas}
\Crefname{lemma}{Lemma}{Lemmas}

\usepackage{wrapfig}
\usepackage{float}

\usepackage{color}

\title{Towards convergence to Nash equilibria in two-team zero-sum games}

\author{\and\and Fivos Kalogiannis\\
UC Irvine\\
\and
Ioannis Panageas \\
UC Irvine
\and\and\and
Emmanouil V. Vlatakis-Gkaragkounis \\
Columbia University
}

\date{April 2023}

\usepackage{mwe}
\newcommand{\tabfigure}[2]{\raisebox{-.5\height}{\includegraphics[#1]{#2}}
}

\begin{document}



\maketitle

\begin{abstract}
Contemporary applications of machine learning in two-team e-sports and the superior expressivity of multi-agent generative adversarial networks raise important and overlooked theoretical questions regarding optimization in two-team games. Formally, two-team zero-sum games are defined as multi-player games where players are split into two competing sets of agents, each experiencing a utility identical to that of their teammates and opposite to that of the opposing team. We focus on the solution concept of Nash equilibria (NE). We first show that computing NE for this class of games is \textit{hard} for the complexity class ${\CLS}$. To further examine the capabilities of online learning algorithms in games with full-information feedback, we propose a benchmark of a simple ---yet nontrivial--- family of such games. These games do not enjoy the properties used to prove convergence for relevant algorithms. In particular, we use a dynamical systems perspective to demonstrate that gradient descent-ascent, its optimistic variant, optimistic multiplicative weights update, and extra gradient fail to converge (even locally) to a Nash equilibrium. On a brighter note, we propose a first-order method that leverages control theory techniques and under some conditions enjoys last-iterate local convergence to a Nash equilibrium. We also believe our proposed method is of independent interest for general min-max optimization.
\end{abstract}

\section{Introduction}

Online learning shares an enduring relationship with game theory that has a very early onset dating back to the analysis of \textit{fictitious play} by \citep{robinson1951iterative} and Blackwell's \textit{approachability theorem}~\citep{blackwell1956analog}. A key question within this context is whether self-interested agents can arrive at a game-theoretic \textit{equilibrium} in an \textit{independent} and \textit{decentralized} manner with only \textit{limited feedback} from their environment. 
Learning dynamics that converge to different notions of equilibria are known to exist for two-player zero-sum games \citep{robinson1951iterative,arora2012multiplicative,daskalakis2011near}, potential games \citep{monderer1996potential}, near-potential games \citep{anagnostides2022last}, socially concave games~\citep{golowich2020tight}, and extensive form games~\citep{anagnostides2022faster}. We try to push the boundary further and explore whether equilibria ---in particular, Nash equilibria--- can be reached by agents that follow decentralized learning algorithms in two-team zero-sum games.

\textit{Team competition} has played a central role in the development of game theory \citep{marschak1955elements,von1997team,Bacharach99,gold2005introduction}, economics \citep{marschak1955elements,gottinger1974j}, and evolutionary biology \citep{nagylaki1993evolution,nowak2004emergence}. Recently, competition among teams  has attracted the interest of the machine learning community due to the advances that multi-agent systems have accomplished:
\textit{e.g.}, multi-GAN's \citep{hoang2017multi,hardy2019md} for generative tasks, adversarial regression with multiple learners \citep{tong2018adversarial}, or AI agents competing in e-sports (\textit{e.g.,} CTF \citep{jaderberg2019human} or Starcraft \citep{vinyals2019grandmaster}) as well as card games
\citep{moravvcik2017deepstack,brown2018superhuman,bowling2015heads}.

\paragraph{Our class of games.} We turn our attention to \emph{two-team zero-sum games} a quite general class of min-max optimization problems that include bilinear games and a wide range of nonconvex-nonconcave games as well.
In this class of games, players fall into two teams of size $n,m$ and submit their own randomized strategy vectors independently.
We note that the games that we focus on are not restricted to team games in the narrow sense of the term ``team" as we use it in sports, games, and so on; the players play independently and do not follow a central coordinating authority. Rather, for the purpose of this paper, \textit{teams} are constituted by agents that merely enjoy the same utility function. This might already hint that the solution concept that we engage with is the \textit{Nash equilibrium} (NE).
Another class of games that is captured by this framework is the class of \textit{adversarial potential games.} In these games, the condition that all players of the same team experience the same utility is weakened as long as there exists a \textit{potential function} that can track differences in the utility of each player when they unilaterally deviate from a given strategy profile (see \Cref{sec:adv-potential-games} for a formal definition). 
A similar setting has been studied in the context of nonatomic games \citep{babaioff2009congestion}.


\paragraph{Positive duality gap.}
In two-player zero-sum games, \textit{i.e.},  $n = m = 1$, min-max (respectively max-min) strategies are guaranteed to form a Nash equilibrium due to Von Neumann's minmax theorem \citep{von1928annalen}; ultimately endowing the game with a unique value.
The challenges arise for the case of $n,m>1$; \citet{schulman2019duality} prove that, in general, two-team games do not have a unique value. They do so by presenting a family of team games with a positive duality gap, together with bounds concerning this gap. These bounds quantify the effect of exchanging the order of commitment to their strategy either between the teams as a whole or the individual players.



\paragraph{Solution concept.}
In this work, we examine the solution concept of Nash equilibrium (NE). Under a Nash equilibrium, no player can improve their utility by unilaterally deviating. The main downside of a NE for team games is the fact that such an equilibrium can be arbitrarily suboptimal for the team \citep{basilico2017computing}. 

This is one of the reasons that the solution concept of team-maxmin equilibrium with a coordination device (TMECor) has dominated contemporary literature of team games, especially in regard to applications \citep{farina2018ex,zhang2020computing,cacciamani2021multi}. Under a TMECor, players are allowed to communicate before the game and decide upon combinations of strategies to be played during the game using an external source of randomness.

The undeniable advantage of a TMECor is that the expected utility of the team under it is greater than the expected utility under a NE~\citep{basilico2017computing}. Nevertheless, this favorable property of TMECor can by no means render the study of NE irrelevant. In fact, the study of NE is always of independent interest within the literature of algorithmic game theory ---especially questions corresponding to computational complexity. Moreover, there exist settings in which \textit{ex ante} coordination cannot be expected to be possible or even sensible; for example in 
\begin{inparaenum}[(i)]
    \item environments where the external sources of randomness are unreliable or nonexistent or visible to the adversarial team,
    \item games in which players cannot know in advance who they share a common utility with,
    \item \textit{adversarial potential games}.
\end{inparaenum}
These games can model naturally occurring settings such as
\begin{inparaenum}[(a)]
    \item security games with multiple uncoordinated defenders versus multiple similarly uncoordinated attackers,
    \item the load balancing ``game'' between telecommunication service providers trying to minimize the maximum delay of service experienced by their customers versus the service users that try to individually utilize the maximum amount of broadband possible, and
    \item the \textit{weak selection} model of evolutionary biology where a species as a whole is a team, the genes of its population are the players and the alleles of each gene are in turn the actions of a player; the allele frequencies are independent across genes \citep{nagylaki1993evolution,nowak2004emergence,MPP15}. 
\end{inparaenum}

Concluding, we could not possibly argue for a single correct solution concept for two-team games; there is no silver bullet. In contrast, one has to assess which is the most fitting based on the constraints of a given setting. A Nash equilibrium is a cornerstone concept of game theory and examining its properties in different games is always important.

\paragraph{The optimization point of view.} 
We focus on the solution concept of NE and we first note that computing local-NE in general nonconvex-nonconcave games is $\PPAD$-complete \citep{daskalakis2009complexity,daskalakis2021complexity}. Thus, all well-celebrated online learning, first-order methods like gradient descent-ascent \citep{lin2020gradient,daskalakis2019last}, its optimistic \citep{popov1980modification,chiang2012online,sridharan2010convex}, optimistic multiplicative weights update \citep{sridharan2012learning}, and the extra gradient method \citep{korpelevich1976extragradient} would require an exponential number of steps in the parameters of the problem in order to compute an approximate NE under
the  oracle  optimization  model of \citep{nemirovskij1983problem}. Additionally, in the continuous time regime, similar classes of games exhibit behaviors antithetical to convergence like cycling, recurrence, or chaos \citep{DBLP:conf/nips/Vlatakis-Gkaragkounis19a}.
Second, even if a regret notion within the context of team-competition could be defined, no-regret dynamics are guaranteed to converge only to the set of coarse correlated equilibria (CCE) \citep{fudenberg1991jean,hannan20164}. CCE is a weaker equilibrium notion whose solutions could potentially be exclusively supported on strictly dominated strategies, even for simple symmetric two-player games (See also \citep{viossat2013no}).

Surely, the aforementioned intractability remarks for the general case of nonconvex-nonconcave min-max problems provide a significant insight. But, they cannot \textit{per se} address the issue of computing Nash equilibria when the game is equipped with a particular structure, \textit{i.e.}, that of two-team zero-sum games. In fact, our paper addresses the following questions:
\begin{center}
    \textit{
        Can we get provable convergence guarantees to NE of decentralized first-order methods in two-team zero-sum games? 
    }
\end{center}

\paragraph{Our results.} 
    First, with regards to computational complexity, we prove that computing an
    approximate (and possibly mixed) NE in two-team zero-sum games is $\CLS$-hard (\Cref{thm:hardness}); \textit{i.e.}, it is computationally harder than finding pure NE in a congestion game or computing an approximate fixed point of gradient descent.
    
    Second, regarding online learning for equilibrium computation, we prove that a number of established, decentralized, first-order methods are not fit for the purpose and fail to converge even asymptotically. Specifically, we present a simple ---yet nontrivial--- family of two-team zero-sum games (with each team consisting of two players) where projected gradient descent-ascent (GDA), optimistic gradient descent-ascent (OGDA), optimistic multiplicative weights update (OMWU), and the extra gradient method (EG) fail to locally converge to a mixed NE (\Cref{thm:failure}). More broadly, in the case of GDA in nondegenerate team games with unique mixed NE, one could acquire an even stronger result for any high-dimensional configuration of actions and players  (\Cref{thm:unstable}). To the best of our knowledge, the described family of games is the first-of-its-kind in which all these methods provably fail to converge at the same time.
    
    Third, we propose a novel first-order method inspired by adaptive control \citep{bazanella1997,hassouneh2004washout}. In particular, we use a technique that manages to stabilize unstable fixed points of a dynamical system without prior knowledge of their position and without introducing new ones. It is important to note that this method is a modification of GDA that uses a stabilizing feedback which maintains the decentralized nature of GDA.
    
    Finally, in \Cref{sec:experiments} we provide a series of experiments in simple two-team zero-sum GAN's. We also show that multi-GAN architectures achieve better performance than single-agent ones, relative to the network capacity when they are trained on synthetic or real-world datasets like CIFAR10.

\paragraph{{Existing algorithms for NE in multiplayer games.}}
{The focus of the present paper is examining algorithms for the setting of \textit{repeated games} \citep[Chapter 7]{CesaBianLugo99}. If we do not restrict ourselves to this setting, there are numerous centralized algorithms \citep{lipton2003playing,berg2017exclusion} and heuristics \citep{gemp2021sample} that solve the problem of computing Nash equilibria in general multi-player games. 
}

\section{Preliminaries}
\label{sec:prelims}
\paragraph{Our setting.}
A \textit{two-team game} in normal form is defined as a tuple $\Gamma(\calN, \calM,\calA,
\calB,{\{U_A, U_B\}})$. The tuple is defined by\vspace{-0.6em}
\begin{enumerate}[(i)]
    \item a finite set of $n = |\calN|$ \textit{players} belonging to team $A$,
    as well as a finite set of $m = |\calM|$ \textit{players} belonging to team $B$;\vspace{-0.3em}
    \item a finite set of \textit{actions} (or \textit{pure strategies}) $\calA_i=\{\alpha_1, \ldots, \alpha_{n_i}\}$ per player $i\in\calN$;
    where $\calA:=\prod_i\calA_i$ denotes the ensemble of all possible action profiles of team $A$, and respectively, a finite set of \textit{actions} (or \textit{pure strategies}) $\calB_i=\{\beta_1, \ldots, \beta_{n_i}\}$ per player $i\in\calM$, where $\calB:=\prod_i\calB_i$.\vspace{-0.25em}
    \item { a utility function for team $A$, $U_A:\calA\times\calB\to \R$ (resp. $U_B:\calA\times\calB\to \R$ for team $B$)}
\end{enumerate}
\vspace{-0.6em}
 
We also use $\vec{\alpha}=(\alpha_1,\ldots,\alpha_n)$ to denote the strategy profile of team $A$ players and $\vec{\beta}=(\beta_1,\ldots,\beta_m)$ the strategy profile of team $B$ players. 

Finally, each team's \textit{payoff} function is denoted by $U_{A},U_{B}:\calA\times\calB\to\R$, where  the \textit{individual utility} of a player  is identical to her teammates, i.e., $U_{i}=U_{A} \ \&\  U_{j}=U_{B} \ \forall i\in\calN$ and $j  \in \calM$. In this general context, players could also submit \textit{mixed strategies}, i.e, probability distributions over actions.
Correspondingly, we define the product distributions $\vec{x} = (\vec{x}_1,\ldots, \vec{x}_n)$, $\vec{y} = (\vec{y}_1, \ldots, \vec{y}_m)$ as team $A$ and $B$'s strategies respectively, in which $\vec{x}_i \in \Delta(\calA_i)$ and $\vec{y}_j \in \Delta(\calB_j)$. Conclusively, we will write $\calX:=\prod_{i\in\calN}\calX_i=\prod_{i\in\calN}\Delta(\calA_i),\calY:=\prod_{i\in\calM}\calY_i=\prod_{i\in\calM}\Delta(\calB_i)$ the space of mixed strategy profiles of teams $A,B$. A two-team game is called \textit{two-team zero-sum} if $U_{B}=-U_{A} = U$ which is the main focus of this paper. {Moreover, we assume that the game is \emph{succinctly representable} and satisfies the \emph{polynomial expectation property}~\citep{daskalakis2006game}. This means that given a mixed strategy profile, the utility of each player can be computed in polynomial time in the number of agents, the sum of the number of strategies of each player, and the bit number required to represent the mixed strategy profile.}

A \textit{Nash equilibrium} (NE) is a strategy profile
$(\vec{x}^*,\vec{y}^*)\in \calX \times \calY$ such that
\begin{equation}
\left\{
\begin{array}{ll}
    U(\vec{x}^*,\vec{y}^*)\le U(\vec{x}_{i},\vec{x}^*_{-i},\vec{y}^*), ~\forall \vec{x}_i \in \calX_i \\
    U(\vec{x}^*,\vec{y}^*)\ge U(\vec{x}^*,\vec{y}_{j},\vec{y}^*_{-j}), ~\forall \vec{y}_j \in \calY_j
    \end{array}
\right.
\tag{NE}
\footnote{
    We are using here the shorthand $\vec{x}_{-i}$ (or $\vec{y}_{-i}$ ) to highlight the strategy of all agents $\calN$ (or $\calM$) but $i$. }
\end{equation}%


\paragraph{A first approach to computing NE in Two-Team Zero-Sum games.}
Due to the multilinearity of the utility and the existence of a duality gap, the linear programming method used in two-player zero-sum games cannot be used to compute a Nash equilibrium.
For the goal of computing Nash equilibrium in two-team zero-sum games, we have experimented with a selection of online learning, first-order methods that have been utilized with varying success in the setting of {the} two-person zero-sum case. Namely, we analyze the following methods:
\begin{inparaenum}[(i)]
    \item gradient descent-ascent (GDA)
    \item optimistic gradient descent-ascent (OGDA)
    \item extra gradient method (EG)
    \item optimistic multiplicative weights update method (OMWU).
\end{inparaenum}
For their precise definitions, we refer to %
\Cref{sec:gamedynamics}.

The below folklore fact will play a key role hereafter.
\begin{fact}
    Any fixed point of the aforementioned discrete-time dynamics (apart from OMWU)  on the utility function necessarily corresponds to a Nash Equilibrium of the game. \label{rem:fixedpointsNash}
\end{fact}

\noindent 
Hence, an important test for the asymptotic behavior of GDA, OGDA, EG, and OMWU methods is to examine whether these methods stabilize around their fixed points which effectively constitute the Nash equilibria of the game. In ~\Cref{sec:fom-fail}, we show that in the absence of pure Nash equilibria, all the above methods fail to stabilize on their fixed points even for a simple class of two-team game with $(n=2,m=2)$. Consequently, they fail to converge to the mixed Nash equilibrium of the game. 

The presence of these results demonstrates the need for a different approach that lies outside the scope of traditional optimization techniques. Inspired by the applications of washout filters to stabilize unknown fixed points and the adaptive control generalizations of the former, we design a new variant of GDA ``vaned'' with a feedback loop dictated by a pair of two matrices. In contrast to the aforementioned conventional methods, our proposed technique surprisingly accomplishes asymptotic last-iterate convergence to its fixed point, \textit{i.e.}, the mixed Nash equilibria of the team game.
\paragraph{$(\mat{K},\mat{P})$-vaned GDA Method.}
After concatenating the vectors of the minimizing and the maximizing agents --- $\vec{z}^{(k)} = (\vec{x}^{(k)}, \vec{y}\step{k})$--- our method for appropriate matrices $\mat{K}, \mat{P}$ reads:
\begin{gather}\label{eq:KPV}
    \begin{cases}
        \vec{z}\step{k+1} =  \projoperZ{} \Big\{ \vec{z}\step{k} + \eta \binom{ - \nabla_{\vec{x}} f( \vec{z}\step{k} ) } {\nabla_{\vec{y}} f( \vec{z}\step{k} )  } + \eta \mat{K} ( \vec{z}\step{k} - \vtheta\step{k} ) \Big\} \\
        \vtheta\step{k+1} =  \projoperZ{} \Big\{  \vtheta\step{k} + \eta \mat{P} ( \vec{z}\step{k} - \vtheta\step{k} ) \Big\}
    \end{cases}
    \tag{{KPV-GDA}}
\end{gather}  
Intuitively, the additional variable $\vtheta\step{k}$ holds an estimate of the fixed point, and through the feedback law $\eta \mat{K} ( \vec{z}\step{k} - \vtheta\step{k} ) $ the vector $\vec{z}$ stabilizes around that estimate which slowly moves towards the real fixed point of the GDA dynamic. 

\subsection{Two Illustrative Examples}\label{sec:illustrative}
Our first example plays a dual role: first, it demonstrates how two-team min-max competition can {capture} the formulation of multi-agent GAN architectures; second, it hints at the discrepancy between the results of optimization methods, since ---as we will see--- GDA will not converge to the Nash equilibrium/ground-truth distribution. {Generally, the solution that is sought after is the $\min\max$ solution of the objective function \citep{GAN14} which are \NP-hard to compute in the general case \citep{borgs2008myth}; nevertheless, applications of GANs have proven that first-order stationary points of the objective function suffice to produce samples of very good quality.}
\subsubsection{Learning a mixture of gaussians with multi-agent GAN's}
    Consider the case of $\calO$, a mixture of gaussian distribution with two components, 
    $C_1 \sim \mathcal{N}(\vmu, \mI)$ and $C_2 \sim \mathcal{N}(-\vmu, \mI)$ and mixture
    weights $\pi_1, \pi_2$ to be positive such that $\pi_1 + \pi_2 = 1$ and $\pi_1, \pi_2 \neq \frac{1}{2}$.
    
    To learn the distribution above, we utilize an instance of a \textit{Team}-WGAN in which there exists a generating team of agents $G_p: \mathbb{R} \rightarrow \mathbb{R}, G_{\vtheta}:\mathbb{R}^n\rightarrow\mathbb{R}^n$, and a discriminating team of agents $D_{\vec{v}}:\mathbb{R}^n\rightarrow\mathbb{R}, D_{\vec{w}}:\mathbb{R}^n\rightarrow \mathbb{R}$, all described by the following equations:
    \begin{equation}
        \arraycolsep=1.0pt
        \begin{array}{cccccccc}
             \text{Generators:}&G_p (\zeta) & = &  p + \zeta  &,&  G_{\theta}(\vec{\xi}) & =& \vec{\xi} + \vtheta\\
             \text{Discriminators:}&D_{\vec{v}}(\vec{y})& = &\langle \vec{v}, \vec{y} \rangle &,& D_{\vec{w}}(\vec{y})& = &\sum_i w_i y_i^2\\ 
        \end{array}
        \label{eq:teamwgan}
    \end{equation}
    
    The generating agent $G_{\theta}$ maps random noise $\vec{\xi} \sim \mathcal{N}(0, \mI)$ to samples while generating agent $G_p(\zeta)$, utilizing an independent source of randomness $\zeta \sim \mathcal{N}(0,1)$, probabilistically controls the sign of the output of the generator $G_\theta$. The probability of ultimately generating a sample $\vec{y} = \vec{\xi} + \vtheta $ is {in expectation} equal to $p$, while the probability of the sample being $ \vec{y} = - \vec{z} - \vtheta $ is equal to $1 - p$.

    On the other end, there stands the discriminating team of $D_{\vec{v}}, D_{\vec{w}}$. Discriminators, $D_v(\vec{y}), D_w(\vec{y})$ map any given sample $\vec{y}$ to a scalar value accounting for the ``realness'' or ``fakeness'' of it -- negative meaning fake, positive meaning real. The discriminators are disparate in the way they measure {the} realness of samples as seen in their definitions.
    
    

    
    We follow the formalism of the Wasserstein GAN to form the optimization objective:
    \begin{align}
    \arraycolsep=0.5pt
    \begin{array}{cc}
    \displaystyle
      \max\limits_{ \vec{v}, \vec{w}} \min\limits_{\vtheta, p} \Bigg\{ \mathbb{E}_{ \vec{y} \sim \calO } \Big[ D_{\vec{v}} (\vec{y} ) +D_{\vec{w}}( \vec{y} ) \Big] - 
        \displaystyle
        \mathbb{E}_{
              \stackrel{\scriptstyle \vec{\xi} \sim \mathcal{N}(0,\mI),}{\scriptstyle \zeta \sim \mathcal{N}(0,1)} }
        \left[
                \begin{array}{cc}
                    \scriptstyle G_p(\zeta)\cdot
                    \Big( D_{\vec{v}}\big(G_{\vtheta}( \vec{y})\big)+D_{\vec{w}}\big(G_{\vtheta}( \vec{y})\big) \Big)
                    \\ {+} \\
                    \scriptstyle \big(1-G_p(\zeta)\big)\cdot \Big(D_{\vec{v}}\big( {-} G_{\vtheta}( \vec{y}) \big) +D_{\vec{w}}\big( {-} G_{\vtheta}( \vec{y}) \big) \Big)
                \end{array}
                \right]    \Bigg\}
    \end{array}
    \label{eq:teamwganobj}    
    \end{align}
    Equation \Cref{eq:teamwganobj} yields the simpler form:
    \begin{equation}
    \textstyle
         \max\limits_{\vec{v}, \vec{w}} \min\limits_{\scriptstyle \vtheta, p}
            ( \pi_1 - \pi_2 ) \vec{v}^T \vmu - 2 p \vec{v}^T \vtheta + \vec{v}^T\vtheta +\sum_i^n w_{i} {(\mu_i^2 - \theta_i^2) }    \label{eq:teamwganobj_simplified}
    \end{equation}

    \begin{figure}
        \centering
        \includegraphics[width=0.5\textwidth]{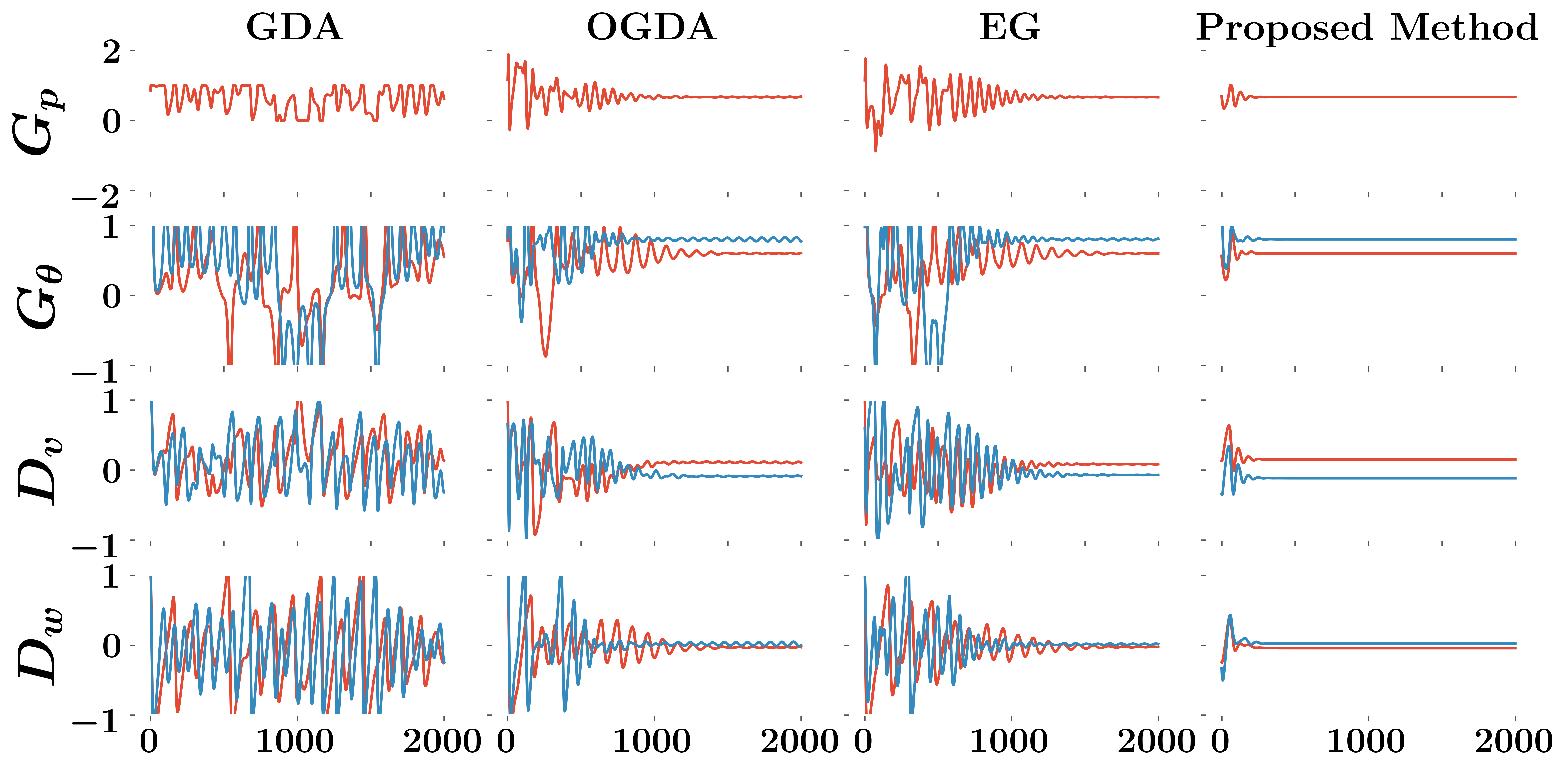}
        \caption{\small Parameter training of the configuration under different algorithms}
        \label{fig:MWGAN-1}
        \vspace{-0.2cm}
    \end{figure}

    It is easy to check that Nash equilibria of \Cref{eq:teamwganobj} must satisfy:
    \[
    \begin{Bmatrix} 
    \vtheta &=& {\color{white}-}\vmu, & p = 1 - \pi_2 = \pi_1 \\
    \vtheta &=& -\vmu, & p = 1 - \pi_1 = \pi_2. 
    \end{Bmatrix}
    \]
    \Cref{fig:MWGAN-1} demonstrates both GDA's failure and OGDA, EG, and our proposed method, {KPV-GDA} succeeding to converge
    to the above Nash equilibria and {simultaneously} discovering the mixture of the ground-truth.
    
\subsubsection{Multiplayer Matching Pennies}
    \label{gmp:example}
    Interestingly enough, there are non-trivial instances of two-team competition settings in which even OGDA and EG fail to converge. Such is the case for a team version of the well-known game of matching pennies. The game can be shortly described as such: \autakia{\textit{coordinate with your teammates to play a game of matching pennies against the opposing team, coordinate not and pay a penalty}}.
    (We note that this game is a special case of the game presented in \Cref{sec:gmp}.)
     As we can see in \Cref{fig:other-foms-fail,fig:allothersucks_phase}, this multiplayer generalized matching pennies game constitutes an excellent benchmark on which all traditional gradient flow discretizations fail under the perfect competition setting. Interestingly, we are not aware of {a} similar example in min-max literature and it has been our starting point for seeking new optimization techniques inspired by Control theory. Indeed, the {KPV-GDA} variation with $(\mat{K},\mat{P})=(-1.1\cdot \mat{I},0.3\cdot \mat{I})$ achieves to converge to the unique mixed Nash Equilibrium of the game. In the following sections, we provide theorems that explain formally the behavior of the {examined} dynamics.  
\begin{figure}[htbp]
    \centering
    \vspace{-10pt}
    \begin{subfigure}[b]{0.5\linewidth}
    \centering
        \includegraphics[width=.8\textwidth]{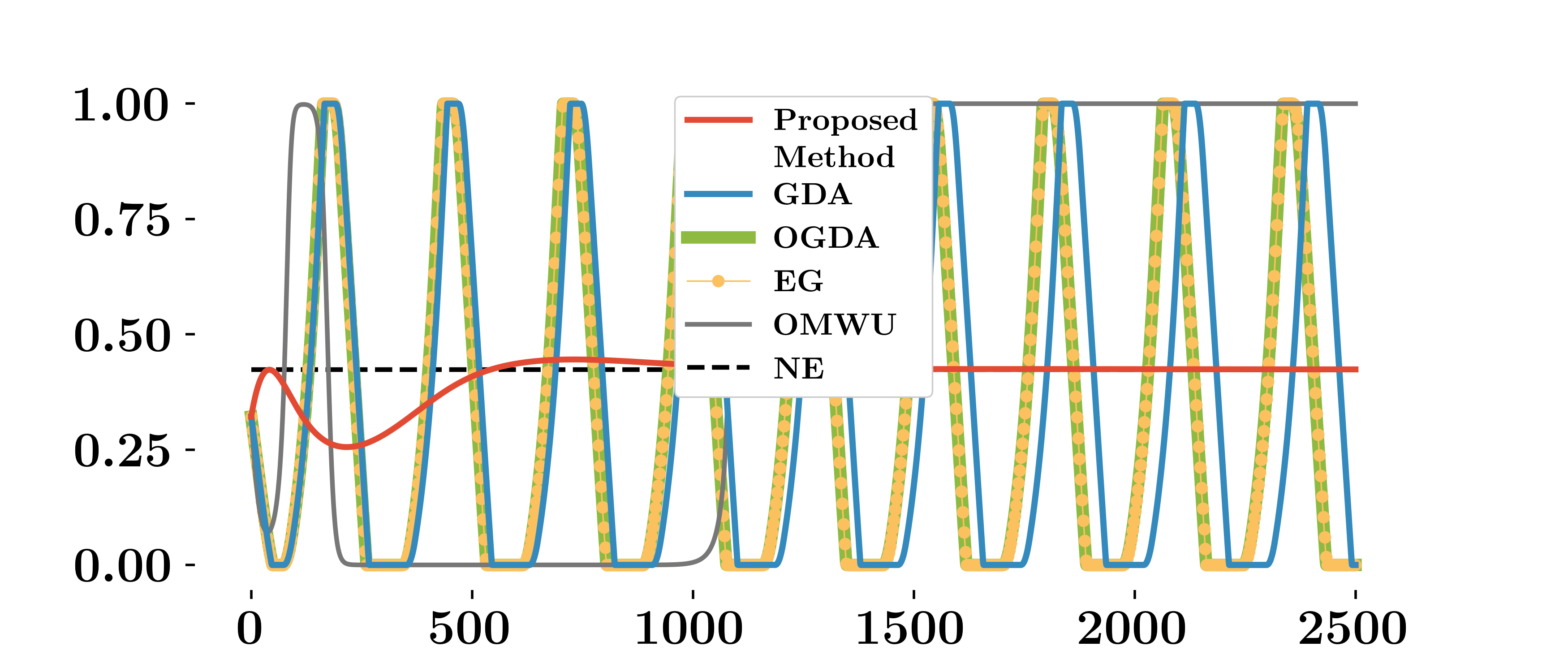}
        \caption{\small Generalized matching pennies under different algorithms. For the precise definition of the game, we refer to appendix%
        ~\ref{sec:tensorMMP}}
         \label{fig:other-foms-fail}  
    \end{subfigure}\hspace{0.5cm}%
    \begin{subfigure}[b]{0.45\linewidth}
    \centering
        \includegraphics[width=0.5\textwidth]{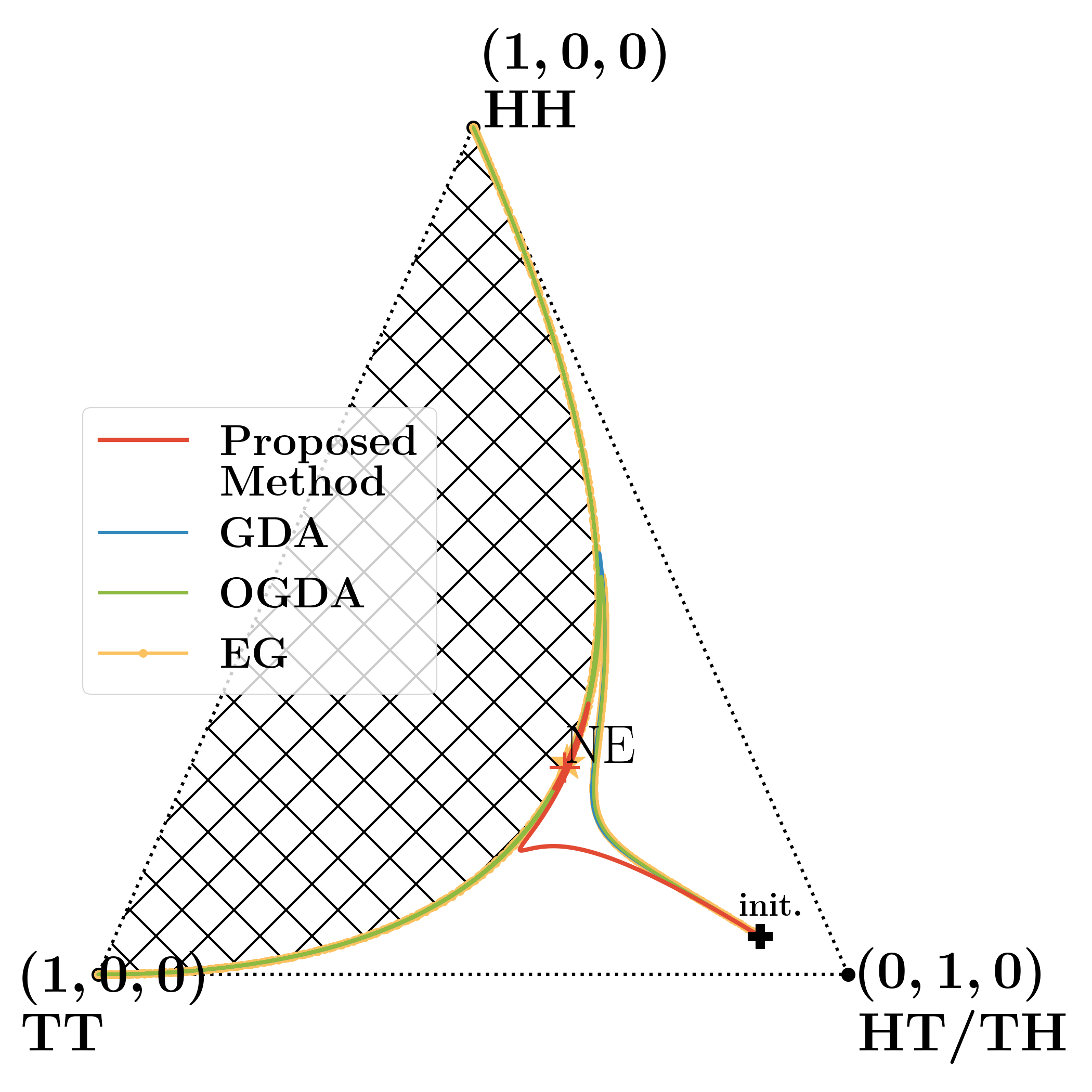}
        \caption{\small Projected Trajectory of Team A under different algorithms.
        The sketched surface is not part of the feasible team strategy profiles (product of distributions).
        }
         \label{fig:allothersucks_phase}  
    \end{subfigure}
    \vspace{-0.5cm}

\end{figure}    
\section{Main results}
\label{sec:main}
In this section, we will prove that
computing a Nash equilibrium in two-team zero-sum games is computationally hard and thus getting a {polynomial-time} algorithm that computes a Nash equilibrium is unlikely. Next, we will demonstrate the shortcomings of an array of commonly used online learning, first-order methods, and then we will provide a novel, decentralized, first-order method that locally converges to NE under some conditions.

\subsection{On the hardness of computing NE in two-team zero-sum games} 
As promised, our first statement characterizes the hardness of NE computation in two-team zero-sum games:

\begin{theorem}[$\CLS$-hard]\label{thm:hardness}
Computing a Nash equilibrium in a {succinctly represented} two-team zero-sum game is $\CLS$-hard. 
\end{theorem} 
The main idea of the proof of \Cref{thm:hardness} relies on a reduction of approximating Nash equilibria in congestion games, which has been shown to be complete for the $\CLS$ complexity class. The class $\CLS$ contains the problem of continuous optimization. We defer the proof of the above theorem to the paper's supplement.

\subsection{Failure of common online, first-order methods} \label{sec:fom-fail}
The negative computational complexity result we proved for two-team zero-sum games (\Cref{thm:hardness}) does not preclude the prospect of attaining algorithms (learning first-order methods) that converge to Nash equilibria. Unfortunately, we prove that these methods cannot guarantee convergence to Nash equilibria in two-team zero-sum games in general. 

In this subsection, we are going to construct a family of two-team zero-sum games with the property that the dynamics of GDA, OGDA, OMWU, and EG are unstable on Nash equilibria. This result is indicative of the challenges that lie in the min-max optimization of two-team zero-sum games and the reason that provable, nonasymptotic convergence guarantees of online learning have not yet been established.

Before defining our benchmark game, we prove an important {theorem} which states that GDA does not converge to mixed Nash equilibria. This fact is a stepping stone in constructing the family of team-zero sum games later. We present the proof {of} all of the below statements in detail in the paper’s appendix
(\Cref{sec:gamedynamics}).

\paragraph{Weakly-stable Nash equilibrium.}{\citep{KPT09, MPP15}} Consider the set of Nash equilibria with the property that if any single randomizing agent of one team is forced to play any strategy in their current support with probability one, all other agents of the same team must remain indifferent between the strategies in their support. This type of Nash equilibria is called weakly-stable. We note that pure Nash equilibria are trivially weakly-stable. It has been shown that mixed Nash equilibria are not weakly-stable in generic games
\footnote{Roughly speaking, generic games where we add small Gaussian noise to perturb slightly every payoff only so that we preclude any payoff ties.
In these games, all  Nash equilibria in all but a measure-zero set of
games exhibit the property that
all pure best responses are played with
positive probability.  
} 

We can show that Nash equilibria that are not weakly-stable Nash are actually unstable for GDA. Moreover, {through standard dynamical systems machinery}, that the set of initial conditions that converges to Nash equilibria that are not weakly-stable should be of Lebesgue measure zero. Formally, we prove that:

\begin{theorem}[Non weakly-stable Nash are unstable]\label{thm:unstable} Consider a two-team zero-sum game with the utility function of Team B ($\vec{y}$ vector) being $U(\vec{x},\vec{y})$ and Team A ($\vec{x}$ vector) being $-U(\vec{x},\vec{y})$. Moreover, assume that $(\vec{x}^*,\vec{y}^*)$ is a Nash equilibrium of full support that is not weakly-stable. It follows that the set of initial conditions so that GDA converges to $(\vec{x}^*,\vec{y}^*)$ is of measure zero for {step size} $\eta < \frac{1}{L}$ where $L$ is the Lipschitz constant of $\nabla U$.
\end{theorem}

\subsection{Generalized Matching Pennies (GMP)}
\label{sec:gmp}
Inspired by \Cref{thm:unstable}, in this section we construct a family of team zero-sum games so that GDA, OGDA, OMWU, and EG methods fail to converge (if the initialization is a random point in the simplex, the probability of convergence of the aforementioned methods is zero). The intuition is to construct a family of games, each of which has only mixed Nash equilibria (that are not weakly-stable), \textit{i.e.}, the constructed games should lack pure Nash equilibria; using \Cref{thm:unstable}, it would immediately imply our claim for GDA. It turns out that OGDA, OMWU, and EG also fail to converge for the same family.

\begin{table}[h]    
\centering
\label{table:pennies}
    \begin{tabular}{cc|c|c|c|}
      & \multicolumn{1}{c}{} & \multicolumn{2}{c}{}\\
      & \multicolumn{1}{c}{} & \multicolumn{1}{c}{$HH$} & \multicolumn{1}{c}{$HT/TH$ } & \multicolumn{1}{c}{$TT$} \\\cline{3-5}
    & $HH$ & $1,-1$ & $\omega,-\omega$ &$-1,1$ \\\cline{3-5}
          & $HT/TH$ & $-\omega,\omega$ & $0,0$ & $-\omega,\omega$ \\\cline{3-5}
      & $TT$ & $-1,1$ & $\omega,-\omega$ & $1,-1$ \\\cline{3-5}
    \end{tabular}
\caption{The generalized matching pennies game (GMP).}
\end{table}

\paragraph{Definition of GMP.} Consider a setting {with} two teams (Team $A$, Team $B$), each of which has $n=2$ players. Inspired by the standard matching pennies game and the game defined in \citep{SV19}, we allow each agent $i$ to have two strategies/actions that is $S = \{H,T\}$ for both teams with $2^{4}$ possible strategy profiles. In case all the members of a Team choose the same strategy say $H$ or $T$ then the Team ``agrees" to play $H$ or $T$ (otherwise the Team ``does not agree"). 


Thus, in the case that both teams ``agree", the payoff of each team is actually the payoff for the {two-player} matching pennies. If one team ``agrees" and the other does not, the team that ``agrees" enjoys the payoff $\omega \in (0,1)$ and the other team suffers a penalty $\omega$. If both teams fail to ``agree", both {teams} get payoff zero. Let $x_{i}$ with $i\in \{1,2\}$ be the probability that agent $i$ of Team $A$ chooses $H$ and $1-x_{i}$ the probability that she chooses $T$. We also denote $\vec{x}$ as the vector of probabilities for Team $A$. Similarly, we denote $y_{i}$ for $i \in \{1,2\}$ be the probability that agent $i$ of Team $B$ chooses $H$ and $1-y_{i}$ the probability that she chooses $T$ and $\vec{y}$ the probability vector.

\paragraph{Properties of GMP.} An important remark on the properties of our presented game is due. Existing literature tackles settings with 
\begin{inparaenum}[(i)]
\item (weak-)monotonocity \citep{mertikopoulos2019optimistic, diakonikolas2021efficient},
\item cocoercivity \citep{zhu1996co}, 
\item zero-duality gap \citep{von1928annalen}, 
\item unconstrained solution space \citep{golowich2020tight}
\end{inparaenum}. Our game is carefully crafted and --although it has a \textit{distinct structure} and is nonconvex-nonconcave only due to \textit{multilinearity}-- satisfies none of the latter properties. This makes the (local) convergence of our proposed method even more surprising. (See also
\Cref{subsec:properties-gmp}.)

The first fact about the game that we defined is that for $\omega \in (0,1)$, there is only one Nash equilibrium $(\vec{x}^*,\vec{y}^*)$, which is the uniform strategy, \textit{i.e.}, $x^*_{1} = x^*_{2} = y^*_{1} = y^*_2= \frac{1}{2}$ for all agents $i.$
\begin{lemma}[GMP has a unique Nash]\label{lemma:onemixed}
The Generalized Matching Pennies game exhibits a unique Nash equilibrium which is $(\vec{x}^*,\vec{y}^*) = ((\frac{1}{2},\frac{1}{2}),(\frac{1}{2},\frac{1}{2})).$
\end{lemma}
\begin{remark}
The fact that the game we defined has a unique Nash equilibrium that is in the interior of $[0,1]^{4}$ is really crucial for our negative convergence results later in the section as we will show that it is not a weakly-stable Nash equilibrium and the negative result about GDA will be a corollary due to \Cref{thm:unstable}. We also note that if $\omega = 1$ then there are more Nash equilibria, in particular the $(\vec{0},\vec{0}),  (\vec{1},\vec{0}), (\vec{0},\vec{1}) ,(\vec{1},\vec{1})$ which are pure.
\end{remark}

The following Theorem is the main (negative) result of this section.

\begin{theorem}[GDA, OGDA, EG{, and OMWU} fail]\label{thm:failure} Consider GMP game with $\omega \in (0,1).$ Assume that $\eta_{\textrm{GDA}}<\frac{1}{4}$, $\eta_{\textrm{OGDA}}< \min( \omega, \frac{1}{8}) $, $\eta_{\textrm{EG}}< \frac{\omega}{2}$ {, and $\eta_{\textrm{OMWU}}<\min\left(\frac{1}{4},\frac{\omega}{2}\right)$ } (bound on the stepsize for GDA, OGDA, OMWU, and EG  methods respectively).
It holds that the set of initial conditions so that GDA, OGDA, OMWU, and EG converge (stabilize to any point) is of measure zero.
\end{theorem}

\begin{remark}
    \Cref{thm:failure} formally demonstrates that the behavior of algorithms mentioned in \Cref{gmp:example} are not a result of ``bad parametrization'', and in fact, the probability that any of them converges to the NE is equal to the probability that the initialization of the variables coincides with the NE (Lebesgue measure zero).
\end{remark}


\begin{remark}[Average iterate also fails]\label{rem:cceneqne}
One might ask what happens when we consider average iterates instead of {the} last iterate. It is a well-known fact \citep{syrgkanis2015fast} that the average iterate of no-regret algorithms converges to  {coarse} correlated equilibria (CCE) so we expect that the average iterate stabilizes. Nevertheless, CCE might not be Nash equilibria. Indeed we can construct examples in which the average iterate of GDA, OGDA, OMWU, and EG experimentally fail to stabilize to Nash equilibria. In particular{,} we consider a slight modification of GMP; players and strategies are the same but the payoff matrix has changed and can be found reads

\end{remark}

\begin{figure}[h!]
\centering
\scalebox{1.1}{
\begin{tabular}{cc}
\begin{tabular}{cc|c|c|c|}
      & \multicolumn{1}{c}{} & \multicolumn{2}{c}{}\\
      & \multicolumn{1}{c}{} & \multicolumn{1}{c}{$HH$} & \multicolumn{1}{c}{$HT/TH$ } & \multicolumn{1}{c}{$TT$} \\\cline{3-5}
    & $HH$ & $2,-2$ & $\frac{1}{2},-\frac{1}{2}$ &$-2,2$ \\\cline{3-5}
          & $HT/TH$ & $-\frac{1}{2},\frac{1}{2}$ & $0,0$ & $-\frac{1}{2},\frac{1}{2}$ \\\cline{3-5}
      & $TT$ & $-1,1$ & $\frac{1}{2},-\frac{1}{2}$ & $1,-1$ \\\cline{3-5}
 \end{tabular}
&    \tabfigure{width=0.5\textwidth}{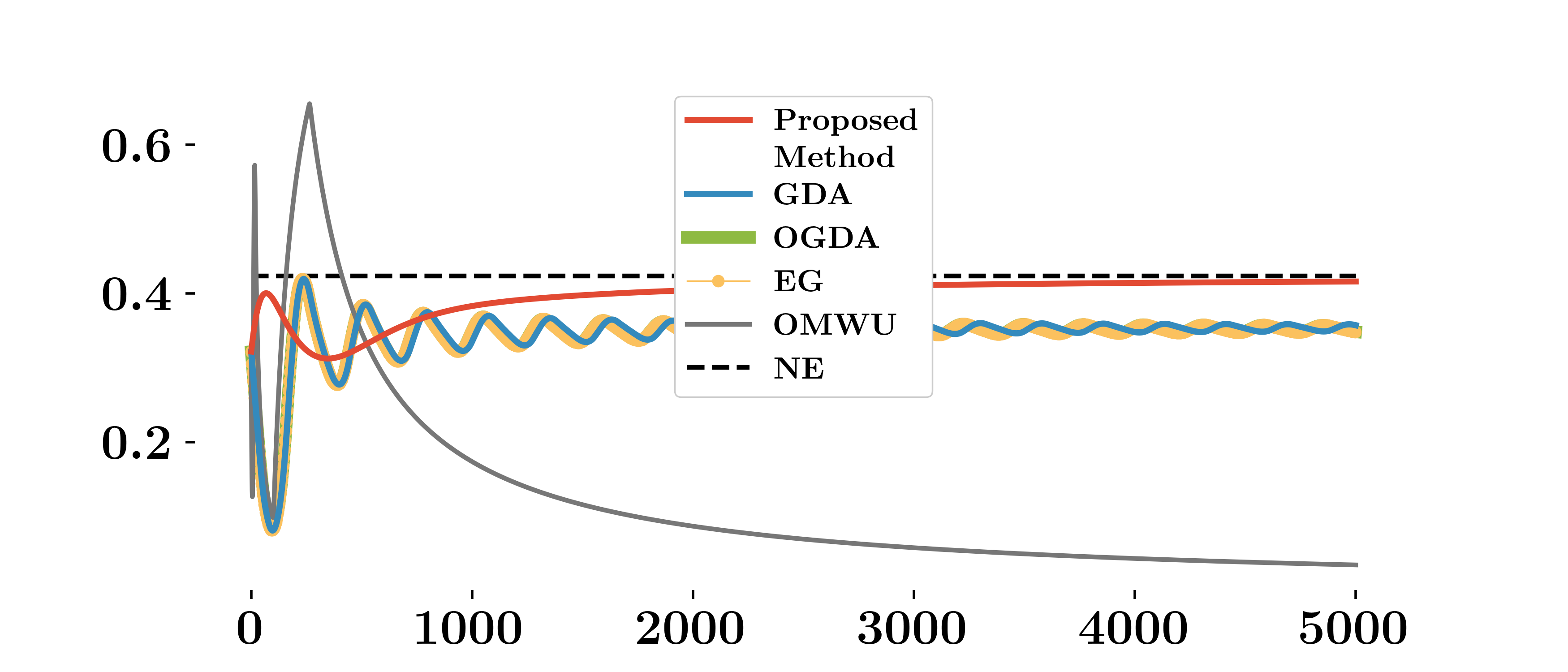}
\end{tabular}
}
\caption{\small {GDA, OGDA, OMWU, \& EG fail to converge to a Nash Equilibrium even in average}}
\label{fig:cycles}
\end{figure}

    

\Cref{fig:cycles} shows that the average iterates of GDA, OGDA, OMWU, and EG {stabilize} to points that are not Nash equilibria. We note that since our method (see next subsection) converges locally, the average iterate should converge locally to a Nash equilibrium.

\subsection{Our proposed method} \label{sec:washout-gda}
The aforementioned results prove that the challenging goal of computing two-team zero-sum games calls for an expansion of existing optimization techniques. The mainstay of this effort and our positive result is the {KPV-GDA} method defined in \Cref{eq:KPV} which is inspired by techniques of adaptive control literature. The first statement we make is that  {KPV-GDA} stabilizes around any Nash {equilibrium} for appropriate choices of matrices $\mat{K},\mat{P}$:


\begin{theorem}[{KPV-GDA} stabilizes]\label{thm:success} Consider a team zero-sum game so that the utility of Team $B$ is $U(\vec{x},\vec{y})$ and hence the utility of Team $A$ is $-U(\vec{x},\vec{y})$ and a Nash equilibrium $(\vec{x}^*,\vec{y}^*)$ of the game. Moreover, we assume 
\[\left( \begin{array}{cc}
-\nabla^{2}_{\vec{x}\vec{x}}U(\vec{x}^*,\vec{y}^*) & -\nabla^{ 2}_{\vec{x}\vec{y}}U(\vec{x}^*,\vec{y}^*)\\
\nabla^{ 2}_{\vec{y}\vec{x}}U(\vec{x}^*,\vec{y}^*) &\nabla^{ 2}_{\vec{y}\vec{y}}U(\vec{x}^*,\vec{y}^*)
\end{array}\right) \textrm{ is invertible.}
\]
For any fixed {step size} $\eta>0$, we can always find matrices $K,P$ so that KPV-GDA method defined in \Cref{eq:KPV}
converges locally to $(\vec{x}^*,\vec{y}^*)$. 
\end{theorem}

This is an existential theorem and cannot be generally useful in practice. Further, this dynamic would not be necessarily uncoupled and the design of matrices $\mat{K}$ and $\mat{P}$ could necessitate knowledge of the NE we are trying to compute. Instead, our next statement provides sufficient conditions under which a simple parametrization of matrices $\mat{K},\mat{P}$ results in an uncoupled, converging dynamic:

\begin{theorem}\label{thm:sufficient} Consider a two-team zero-sum game so that the utility of Team $B$ is $U(\vec{x},\vec{y})$, the utility of Team $A$ is $-U(\vec{x},\vec{y})$, and a Nash equilibrium $(\vec{x}^*,\vec{y}^*)$. Moreover, let 
\[
\scriptsize
\mat{H}: = \left( \begin{array}{cc}
-\nabla^2_{\vec{x}\vec{x}}U(\vec{x}^*,\vec{y}^*) & -\nabla^2_{\vec{x}\vec{y}}U(\vec{x}^*,\vec{y}^*)\\
\nabla^2_{\vec{y}\vec{x}}U(\vec{x}^*,\vec{y}^*) &\nabla^2_{\vec{y}\vec{y}}U(\vec{x}^*,\vec{y}^*)
\end{array}\right).\] 
and $E$ be the set of eigenvalues $\rho$ of $\mat{H}$ with real part positive, that is $E =\{\textrm{Eigenvalues of matrix $\mat{H}$,}~\rho:\textrm{Re}(\rho)>0\}$.
We assume that $\mat{H}$ is invertible and moreover
\begin{equation}\label{eq:assumption}
\beta=
{  \min_{\rho\in E} \frac{\textrm{Re}(\rho)^2 +\textrm{Im}(\rho)^2}{\textrm{Re}(\rho)}}
>
{
\max_{\rho\in E} \textrm{Re}(\rho)}
=\alpha.\end{equation}
We set $\mat{K} = k \cdot  \mat{I}, $ $\mat{P} = p \cdot  \mat{I}. $ There exist small enough step size $\eta>0$ and scalar  $p>0$ and for any $k\in (-\beta,-\alpha)$ so that \Cref{eq:KPV} with chosen $\mat{K},\mat{P}$ converges locally to $(\vec{x}^*,\vec{y}^*)$.\footnote{As long as aforementioned conditions are satisfied, \Cref{eq:KPV} locally converges in any nonconvex-nonconcave game. Indeed, GMP with any $\omega$ satisfies the sufficient conditions of \ref{thm:sufficient}. See also, \Cref{sec:gmp-satisfies-sufficient-conds}.}
\end{theorem}




\section{Experiments}\label{sec:experiments}

In this section, we perform a series of experiments to further motivate the study of two-team zero-sum games, especially in the context of multi-agent generative adversarial networks (multi-GANs). A multi-agent generative adversarial network (multi-GAN)~\citep{arora2017generalization,hoang2017multi,hoang2018mgan,zhang2018stackelberg,tang2020lessons,hardy2019md,albuquerque2019multi} is a generative adversarial network (GAN) that leverages multiple  ``agents'' (generators and/or discriminators)  in order achieve statistical and computational benefits. 
In particular, \citeauthor{arora2017generalization} formally proved the expressive superiority of multi-generator adversarial network architectures something that we empirically verify in \Cref{sec:experiments}. 
In this direction, researchers strive to harness the efficacy of distributed processing by utilizing shallower networks that can collectively learn more diverse datasets%
\footnote{
Indeed, it is preferable from a computational standpoint to back-propagate through two equally sized neural networks 
rather than through a single one that would be twice as deep \citep{tang2020lessons}.
}.

At first, the superiority of multi-GANs might appear to contrast our theoretical findings; but in reality, the superiority comes from the quality of solutions that are attainable from multi-agent architectures (\textit{expressivity}) and the fact that hardness (\textit{complexity}) translates to rates of convergence but not non-convergence. Single agent GANs quickly converge to critical points that are not guaranteed to capture the distribution very well. In figure \ref{big-vs-many} we see the fast convergence of a single-agent GAN to solutions of bad quality versus the convergence of a multi-GAN to an obviously better solution.  Due to space constraints, we defer further discussion of the experiments at Section 
\ref{sec:experiments_gans}.

\begin{figure}[h!]
\centering
\vspace{-0.5cm}
\begin{tabular}{ccc}
    \includegraphics[width=0.3\textwidth]{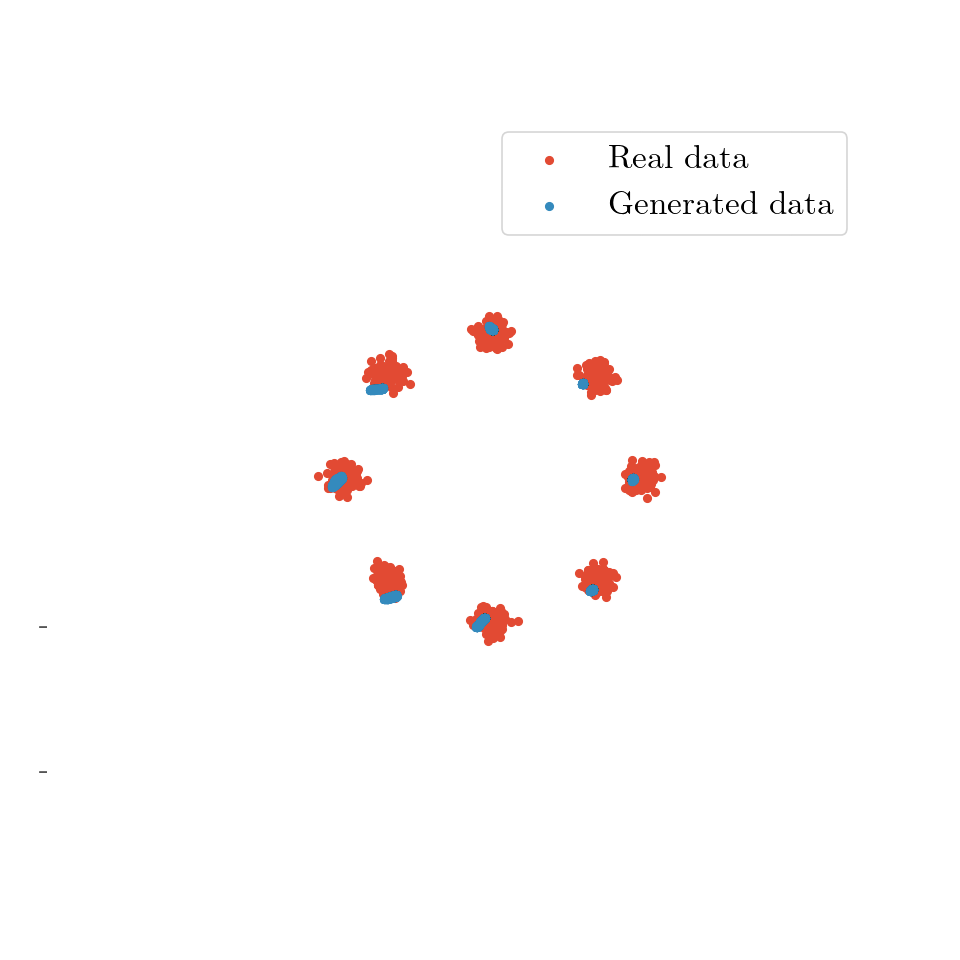} &
    \includegraphics[width=0.3\textwidth]{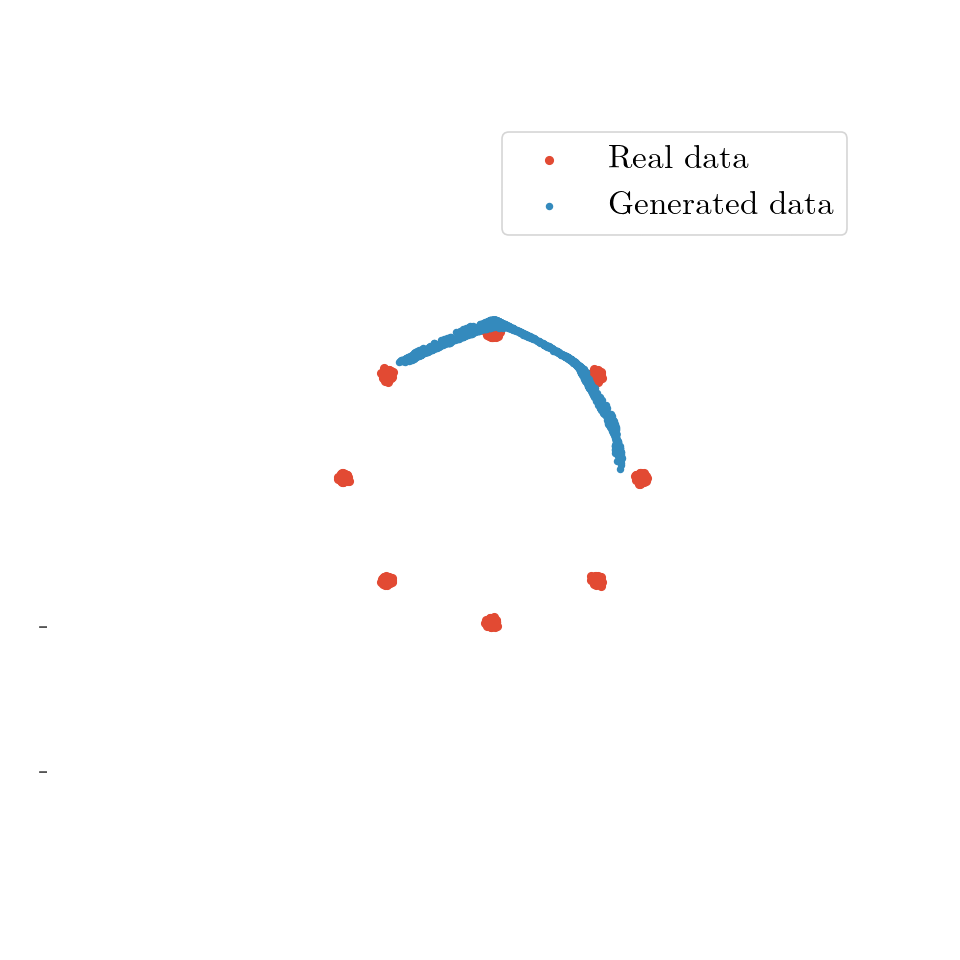} &
    \includegraphics[width=0.3\textwidth]{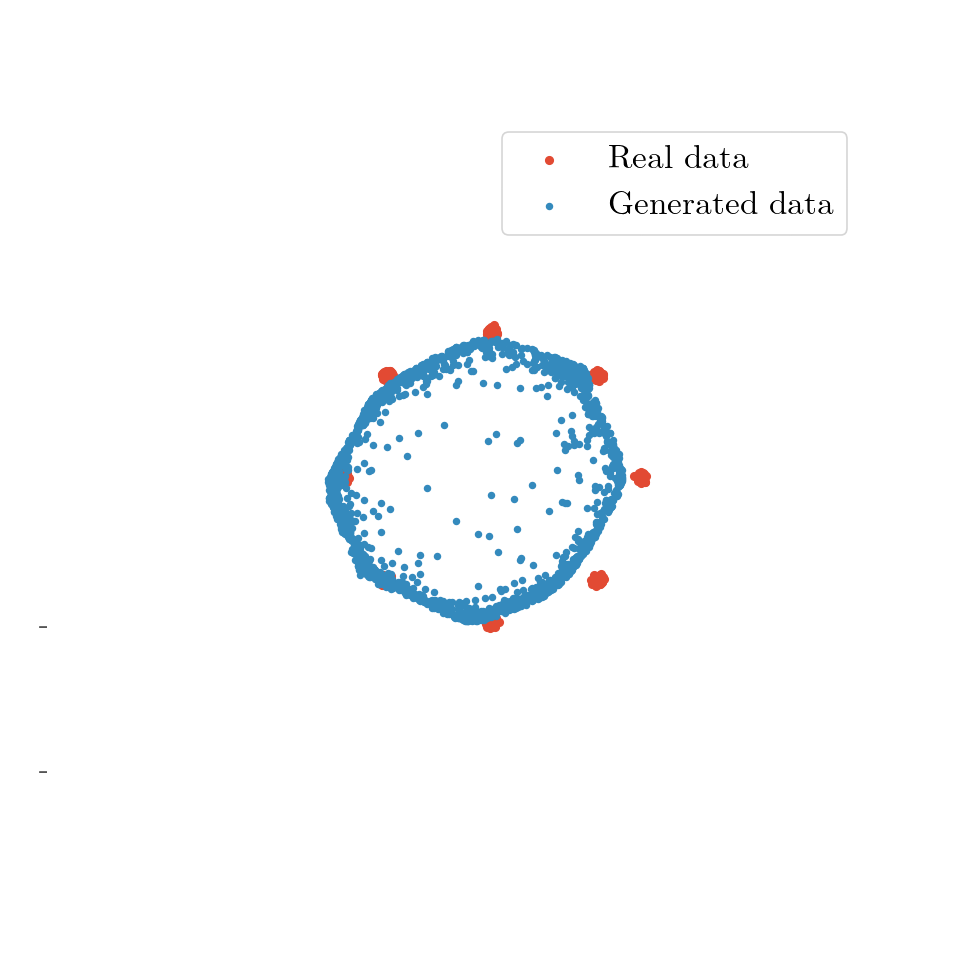} 
\end{tabular}
\vspace{-40pt}
\caption{\small
From left to right: \inlinelist{\protect\item Each generator of MGAN learns one mode of the mixture of 8 gaussians, \protect\item Mode Collapse of single-agent GANs, \protect\item Single-agent GAN can't discriminate between the modes.}
}
\label{big-vs-many}
\end{figure}

\section{Conclusions and Open Problems} \label{sec:conclusion}

In this work we study the wide class of nonconvex-nonconcave games that express the two-team competition, inspired broadly by the structure of the complex competition between multi-agent generators and
discriminators in GAN's.
{Furthermore}, in this setting of \emph{two-team zero-sum games}, we have presented a number of negative results about the problem of computing a Nash equilibrium. Moreover, through a simple family of games that we construct, we prove the inability of commonly used methods for min-max optimization such as GDA, OGDA, OMWU, and EG to converge both in average and in the last iterate to Nash Equilibria which comes to a stark contrast to recent literature that is concerned with simpler games. We have also presented an optimization method (called {KPV-GDA}) that manages to stabilize around Nash equilibria.



\section*{Acknowledgements}
Ioannis Panageas would like to acknowledge a start-up grant. Emmanouil V. VG is grateful for the financial support by FODSI Postdoctoral Fellowship. 
He is also grateful for financial support by the Google-Simons Fellowship, Pancretan Association of America and Simons Collaboration on Algorithms and Geometry. Additionally, he would like to acknowledge the following series of NSF-CCF grants under the numbers 1763970/2107187/1563155/1814873. 
This work was partially completed while IP and EVVG were visiting research fellows at the Simons Institute for Theory of Computing during Learning and Games Semester.

\bibliography{main}

\newpage

\appendix

\section{Further related work}
\label{appendix:sol_concepts}
 The properties of team games have given rise to a large amount of literature
which we cannot hope to review here; instead we refer the reader to \citep{kim2019theory,gold2005introduction} and references therein. Additionally, convergence properties of online learning, first-order methods (\textit{e.g.} GDA, OGDA, OMWU, EG) in \textit{team games} are not well understood.

\paragraph{Team games.} For team games, the solution concept of \textit{team-maxmin equilibria} (TME), a notion introduced by \citet{von1997team}, is a rather well-understood solution concept. In this category of two-team zero-sum games, a team of $n$ players  plays defensively against a single \textit{adversarial} player. TME's provide crucial information about robustness since they  give the  team  players  the  highest  payoff possible when they employ uncoordinated strategies. Unfortunately, TMEs correspond to the global optimum of a nonconvex nonlinear program which is $\NP$-hard to compute and as such precludes the use of online learning, first-order algorithms for their computation. \cite{basilico2017team} provided heuristics with iterative linear programming and support enumeration methods with proved approximation guarantees. For the case of extensive or large-sized normal form games, \citet{basilico2017computing,celli2018computational}  introduced a relaxed notion of \textit{team-maxmin  equilibrium  with  coordination  device} (TMECor). In a TMECor, the team of players decides their joint strategies \emph{ex-ante} against the adversary; that is, the team members are allowed to discuss and agree on tactics before the game starts, but they cannot communicate during the game.

While TMECor's share some properties with NE's in zero-sum two-player games (\textit{e.g}., exchangeability), the proposed algorithms by \cite{celli2018computational} leverage tools like Mixed-Integer Linear Program (MILP) that involve a large number of integer variables. In order to tackle the large size of the program heuristically, \citet{zhang2020computing,zhang2020converging} have proposed modified versions of the incremental  strategy  generation  (ISG)  algorithm (See \cite{mcmahan2003planning}) for {large-sized} games.

\paragraph{Multi-GAN's.} We cannot aspire to survey the literature of multi-GAN's in its entirety. Scratching only the surface, we mention that: 
\begin{inparaenum}[(a)]
\item In MGAN \cite{hoang2018mgan} and MAD-GAN \cite{li2019mad}, we have a mixture of generative models with asymptotically infinite capacity networks. 
\item In the seminal case of MIX+GAN \cite{arora2017generalization}, an extra regularization term is typically added to discourage the weights of a mixture of generators being too far away from uniform. 
\item In Stackelberg GAN, \cite{zhang2018stackelberg} exploit the leader-followers relation of discriminator vs generators to achieve smaller minimax gap and better Frech\'{e}t metrics.
For the case of \textit{multiple-discriminators}: 
\item In GMAN \cite{durugkar2016generative}, an increased number of agents are used to achieve higher quality images, while 
\item  MD-GAN \cite{hardy2019md} is designed to generalize over different datasets and multiple tasks that are spread on multiple workers.
\end{inparaenum}

\subsection{ Discussion about different solution concepts}\label{ap:discussion}
In this subsection, we would like to stress some differences between the different equilibrium notions (TMECor, TME, and NE). Initially, it is easy to see that TMECor similar to coarse correlated equilibria can achieve better social welfare but, as in the case of congestion games, their price of anarchy can become significantly worse \cite{roughgarden2009intrinsic}.
In order to improve this, TMECor request prior knowledge of the game. For example, as \citet{zhang2020computing} mention:
\textit{For example, in multiplayer poker games,
a team may play against an adversary player,
but they cannot communicate and discuss
their strategy during the game due to the rule.}

However in many contemporary challenging AI models, the agents/players of each team do not necessarily have knowledge of the actual game. Thus, the only thing the team players can decide in advance is the decision rules they will follow. Thus, the main conceptual reason that we focus on NE's in comparison to TMECor is is because it is a way of reasoning about games in which players do not necessarily need to know in which team they fall into in advance and/or they cannot be expected to be able to coordinate their actions with players of common interest.

Problems of this kind are prevalent in modern ML applications like  strategic  conflict  resolution  \citep{leibo2017multi},  coordination  between  autonomous  vehicles  \citep{cao2012overview} , or  collaboration of  agents  in  defensive  escort  teams  \citep{sheikh2019learning}. In this kind of games  each  agent  is  simultaneously working  towards  maximizing  its  own  payoff  \textit{(local  reward)} as well as the collective success of the team \textit{(global reward)}.

    

Conclusively , the notion of Nash equilibrium seems to be a solution concept that needs to be investigated thoroughly in the setting of two-team zero-sum games.

\subsection{Adversarial potential games.} 
\label{sec:adv-potential-games}
Consider a normal form game with two sets of players $\calN$ and $\calM$. Following the conventions of \Cref{sec:prelims}, the minimizers $\calN$ will be associated with cost functions $c_i,~i\in [n]$, while the maximizers $\calM$ will be associated with utilities $u_j, ~j\in[m]$. Further, there exists an underlying \emph{potential function} $\Phi : \calX \times \calY \to \R$, so that 
\begin{itemize}[noitemsep]
    \item[(i)] for each minimizer $i \in \calN$ and strategies $\vec{x}_i,\vec{x}_i' \in \calX_i$, 
    \begin{equation}
        \label{eq:cost-pot}
    c_i(\vec{x}_{i}',\vec{x}_{-i},\vec{y}) - c_i(\vec{x}_{i},\vec{x}_{-i},\vec{y}) = \Phi(\vec{x}_{i}',\vec{x}_{-i},\vec{y}) - \Phi(\vec{x}_{i},\vec{x}_{-i},\vec{y});    
    \end{equation}
    \item[(ii)] for each maximizer $j \in \calM$ and strategies $\vec{y}_j,\vec{y}_j' \in \calY_j$,
    \begin{equation}
        \label{eq:util-pot}
        u_j(\vec{x},\vec{y}_{j}',\vec{y}_{-j}) - u_j(\vec{x},\vec{y}_{j},\vec{y}_{-j}) = \Phi(\vec{x},\vec{y}_{j}',\vec{y}_{-j}) - \Phi(\vec{x},\vec{y}_{j},\vec{y}_{-j}).
    \end{equation}
\end{itemize}
While this setting is not a two-team zero-sum game, it can still be captured by one, namely $\Gamma(n,m,\Phi)$. More precisely, \cref{eq:cost-pot} implies that $c_i (\vec{x},\vec{y}) = \Phi(\vec{x},\vec{y}) + \Psi_i(\vec{x}_{-i},\vec{y})$, for any $i \in \calN$ and $(\vx, \vy) \in \calX \times \calY$, where $\Psi_i$ does not depend on $\vec{x}_i$, and \cref{eq:util-pot} implies that $u_j (\vec{x},\vec{y}) = \Phi(\vec{x},\vec{y}) + \Psi'_j(\vec{x},\vec{y}_{-j})$, where $\Psi_j$ does not depend on $\vec{y}_j$. As a result, by the same reasoning we used in the proof of \cref{thm:hardness}, we arrive at the following conclusion:

\begin{proposition}
    A strategy profile that is an $\epsilon$-approximate Nash equilibrium with respect to the two-team zero-sum game $\Gamma(n, m, \Phi)$ will be an $\epsilon$-approximate Nash equilibrium of the original adversarial potential game.
\end{proposition}

{
 
\section{Game Dynamics}

The following algorithms are ubiquitous in the literature of min-max optimization and have known variable success under different settings. Gradient descent-ascent (GDA) is the prototypical method upon which optimistic gradient descent (OGDA), extra gradient (EG), and our proposed method $(\mat{K},\mat{P})$-vaned gradient descent-ascent (KPV-GDA) are built. Optimistic multiplicative weights update is the optimistic variant of the MWU.
\label{sec:gamedynamics}
\subsection{Gradient Descent-Ascent}
    $$ 
    \begin{cases}
        x_i\step{t+1} = \projoperX{i} \Big\{ x_i\step{t} - \eta \nabla_{x_i} f  (\vec{x}\step{t} \vec{y}\step{t} ) \Big\} \\
        y_j\step{t+1} = \projoperY{j} \Big\{ y_j\step{t} + \eta \nabla_{y_j} f  (\vec{x}\step{t} \vec{y}\step{t} ) \Big\}
    \end{cases}$$
\subsection{Optimistic Gradient Descent-Ascent}
    $$ \begin{cases}
        x_i\step{t+1} = \projoperX{i} \Big\{ x_i\step{t} - 2 \eta \nabla_{x_i} f  (\vec{x}\step{t} \vec{y}\step{t} ) +  \eta \nabla_{x_i} f  (\vec{x}\step{(t-1)}), \vec{y}\step{t-1} \Big\} \\
        y_j\step{t+1} = \projoperY{j} \Big\{ y_j\step{t} + 2 \eta \nabla_{y_j} f  (\vec{x}\step{t} \vec{y}\step{t}) -
        \eta \nabla_{y_j} f  (\vec{x}\step{(t-1)}), \vec{y}\step{t-1} \Big\}
    \end{cases}$$
\subsection{Extra Gradient Method}
    $$ \begin{cases}
        x_i\step{t+\frac{1}{2}} = \projoperX{i} \Big\{ x_i\step{t} - \eta \nabla_{x_i} f  (\vec{x}\step{t} \vec{y}\step{t}) \Big\},  \quad 
        x_i\step{t+1} = \projoperX{i} \Big\{ x_i\step{t} - \eta \nabla_{x_i} f  (\vec{x}\step{t+\frac{1}{2}}, \vec{y}\step{t+\frac{1}{2}}) \Big\} \\
        y_j\step{t+\frac{1}{2}} = \projoperY{j}\Big\{ y_j\step{t} + \eta \nabla_{y_j} f  (\vec{x}\step{t} \vec{y}\step{t}) \Big\}, 
        \quad  
        y_j\step{t+1} = \projoperY{j} \Big\{ y_j\step{t} + \eta \nabla_{y_j} f  (\vec{x}\step{t+\frac{1}{2}}, \vec{y}\step{t+\frac{1}{2}}) \Big\}
        \end{cases}
    $$
\subsection{Optimistic Multiplicative Weights Update Method}
    $$ \begin{cases}
        x_{i,k}\step{t+1} =  x_{i,k}\step{t} 
        \frac{\exp\Big( 
            - 2 \eta \frac{\partial}{\partial{x_{i,k}}} f  (\vec{x}\step{t} \vec{y}\step{t} ) +  \eta \frac{\partial}{\partial{x_{i,k}}} f  (\vec{x}\step{(t-1)}), \vec{y}\step{t-1}) \Big) } 
        { 
            \sum_j x_j\step{t} \exp\Big( - 2 \eta \frac{\partial}{\partial{x_{i,k}}} f  (\vec{x}\step{t} \vec{y}\step{t} ) +  \eta \frac{\partial}{\partial{x_{i,k}}} f  (\vec{x}\step{(t-1)}), \vec{y}\step{(t-1)}) \Big) 
        } 
        \\
        y_{j,k}\step{t+1} =  y_{j,k}\step{t} \frac{  \exp\Big( 2 \eta \frac{\partial}{\partial{y_{j,k}}} f  (\vec{x}\step{t} \vec{y}\step{t}) -
        \eta \frac{\partial}{\partial{y_{j,k}}}f  (\vec{x}\step{(t-1)}), \vec{y}\step{t-1} \Big) }
            { \sum_j y_j\step{t} \exp\Big( 2 \eta \frac{\partial}{\partial{y_{j,k}}} f  (\vec{x}\step{t} \vec{y}\step{t}) - \eta \frac{\partial}{\partial{y_{j,k}}}f  (\vec{x}\step{(t-1)}), \vec{y}\step{t-1} \Big)} 
    \end{cases}$$ 
    
\subsection{\texorpdfstring{$\mat{K},\mat{P}$}{K,P}-vaned Gradient Descent-Ascent}
    $$\begin{cases}
        \vec{z}\step{t+1} = & \projoperZ{} \Big\{ \vec{z}\step{t} + \eta \binom{ - \nabla_{\vec{x}} f( \vec{z}\step{t} ) } {\nabla_{\vec{y}} f( \vec{z}\step{t} )  } + \eta \mat{K} ( \vec{z}\step{t} - \vtheta\step{t} ) \Big\} \\
        \vtheta\step{t+1} = & \projoperZ{} \Big\{  \vtheta\step{t} + \eta \mat{P} ( \vec{z}\step{t} - \vtheta\step{t} ) \Big\}
    \end{cases}$$

}

{
Notations $x_i\step{t}$ ( or $y_j\step{t}$) stand for the strategy vector of the $i$-th minimizing (or $\mat{J}$ -th maximizing) agent at time-step $t$; while $x_{i,k}$ and $y_{j,k}$ denote the $k$-th pure strategy of player $i$ of team A and player $\mat{J}$  of team $B$ respectively. The step size is denoted as $\eta$. Operators $\projoperX{i},\projoperY{j}, \projoperZ{}$ are the projection operators to the corresponding products of simplices. The projection operator is important since remaining inside of the set of feasible solutions is not guaranteed at every step except for the case of the OMWU method.
}

\section{Properties of GMP}

\label{subsec:properties-gmp}
We will formally demonstrate how a series of conditions needed for the convergence of the extra gradient method are not satisfied for \hyperref[sec:gmp]{GMP}.

\subsection{Lack of favorable properties}
\paragraph{Variational Inequalities}
Let a saddle-point problem (which can be seen as a zero-sum game between agents $\vec{x}, \vec{y}$):
\begin{equation}
    \min_{\vec{x} \in \calX} \max_{\vec{y} \in \calY} U(\vec{x}, \vec{y})
\end{equation}

 Also, let an operator $F:\R^d \rightarrow \R^d$ that is $L$-Lipschitz continuous for a norm $\| \cdot \|_p$.

The saddle-point problem can be formulated in the form of the equivalent Stampacchia variational inequality (SVI), i.e. find a $\vec{z}^*  \in \R^d$, $\vec{z}^* = (\vec{x}^*, \vec{y}^*)$, such that for every $\vec{z}$:
\begin{equation}
    \langle F(\vec{z}^*) - \vec{z}-\vec{z^*} \rangle \geq 0.
\end{equation}

Operator $F$ is said to be monotone if for every $\vec{z}, \vec{w} \in \R^d$:
\begin{equation}
    \langle F(\vec{z}), F(\vec{w}), \vec{z}-\vec{w} \rangle \geq 0.
\end{equation}

One more noteworthy variational inequality is the Minty Variational Inequality (MVI) with respect to a point $ \vec{z}^* $ is defined as:
\begin{equation}
    \langle F(\vec{z}), \vec{z} - \vec{z}^*  \rangle \geq 0.
\end{equation}
 
\citet{mertikopoulos2019optimistic} prove the convergence of the extra gradient method by utilizing the monotonicity of $F$. In this case, the solutions to the MVI, $\vec{z}^*$, coincide with the solutions to the SVI. 
In particular $F$ is taken to be
$F(\vec{x}, \vec{y}) = \Big( \nabla_{\vec{x}} U(\vec{x}, \vec{y}) , -\nabla_{\vec{y}} U(\vec{x}, \vec{y})\Big)$.

\citet{diakonikolas2021efficient} prove convergence under a weaker form of the MVI for operator $F$:
\begin{equation}
    \langle F(\vec{z}), \vec{z} - \vec{z}^*  \rangle \geq - \frac{\rho}{2} \| F(\vec{z} ) \|_{p^*}^2.
\end{equation}
With $\rho$ a constant in the interval $[0, \frac{1}{4L}]$. We note that with $\rho = 0$ we retrieve the MVI which we proved is not satisfied.

Setting $\vec{z} = (\frac{1}{3}, \frac{2}{3}, \frac{1}{6}, \frac{5}{6}, \frac{2}{3}, \frac{1}{3}, \frac{1}{3}, \frac{2}{3} )$ we see that the MVI is violated: $\langle F(\vec{z}), \vec{z} - \vec{z}^*\rangle = - \frac{\omega}{9}$.

As for the weak-MVI, we numerically verify that for $\omega = 1/2$ and every $\rho$ in $\{ \frac{k}{10^3} | k = 0, 1, ..., 10^3 \cdot \frac{1}{4L} \}$  that a $\vec{z}$ can be found such that is violates the inequality.

Hence, already known techniques \citep{mertikopoulos2018cycles, diakonikolas2021efficient, golowich2020tight} cannot tackle the problem of computing a Nash equilibrium in GMP. Additionally, \citet{golowich2020tight} only analyze the unconstrained setting.

\paragraph{Cocoercivity}
Cocoerciviy is a condition that is stronger than the MVI. Hence, it follows that cocoercivity does not hold in
\hyperref[sec:gmp]{GMP}. 

\paragraph{Hidden convex-concave games.}
\citet{DBLP:conf/nips/FlokasVGP21} proved global convergence for gradient flow (continuous time) when a game is strictly hidden convex-concave. Our setting of \hyperref[sec:gmp]{GMP} does not meet either criterion.

{
\subsection{GMP satisfies the sufficient conditions of \Cref{thm:sufficient}}
\label{sec:gmp-satisfies-sufficient-conds}
We can verify that the generalized matching pennies game satisfies the sufficient conditions of \Cref{thm:sufficient}. Namely, the eigenvalues of $H$ are the following
$$ \rho \in \{ -2\omega, 2\omega - 2i, 2\omega + 2i \}.$$

Hence, the set of the eigenvalues with a positive real part $E$ is
$E = \{ 2\omega - 2i, 2\omega + 2i \}$. As such, we see that indeed,
$$\beta = \min_{\rho \in E} \frac{\mathrm{Re}(\rho) + \mathrm{Im} (\rho) }{\mathrm{Re}(\rho)} = \frac{2\omega^2 + 2}{\omega} = 2 \omega + \frac{2}{\omega} > 2 \omega = \max_{\rho \in E} \mathrm{Re}(\rho) = \alpha .$$

For any choice of $k \in (  -2 \omega - \frac{2}{\omega}, -2\omega)$ the algorithm converges. For the case that $\omega = 1/2$, it suffices that $k \in (-5, -1)$.
}
{ 
\section{Proofs for \cref{sec:illustrative}}
\subsection{Derivation of the min-max objective in ~\eqref{eq:teamwganobj_simplified}}
By the definition of variance, we get that :
$\mathrm{Var} \big[ x_i \big] = \mathbb{E}[{x_i^2}] - \big( \mathbb{E}[{x_i}] \big)^2 \Leftrightarrow \mathbb{E}[{x_i^2}] = \mathrm{Var} \big[ x_i \big] + \big( \mathbb{E}[{x_i}] \big)^2$. More precisely for any $\vec{x} \sim \mathcal{N}{(\mu, \mat{I})}$, we get $\mathbb{E}[{x_i^2}] = 1 +  \mu_i ^2$.

Additionally, after calculations, we can get that:
\begin{align*}
        &\mathbb{E}_{ \xi \sim \mathcal{N}(0,\mI)} [{( \xi_i + \theta_i)^2}] = \mathrm{Var} \big[  \xi_i + \theta_i \big] + \big( \mathbb{E}[{ \xi_i + \theta_i}] \big)^2 \xRightarrow[]{\theta_i \text{ const.} } \\
        &\mathbb{E}_{ \xi \sim \mathcal{N}(0,\mI)} [{( \xi_i + \theta_i)^2}] = \mathrm{Var} \big[  \xi_i \big] + \big( \mathbb{E}[\xi_i] + \mathbb{E}[\theta_i] \big)^2 
        \Rightarrow \\
        &\mathbb{E}_{ \xi \sim \mathcal{N}(0,\mI)} [{( \xi_i + \theta_i)^2}] = \mathrm{Var} \big[  \xi_i \big] + \big( 0 +  \theta_i \big)^2 
        \Rightarrow \\
        &\mathbb{E}_{ \xi \sim \mathcal{N}(0,\mI)} [{( \xi_i + \theta_i)^2}] = 1 +  \theta_i^2 
    \end{align*}

Moreover, it is easy to check that for a mixture $\mathcal{D}(\mu_1, \sigma_1\mI,\mu_2, \sigma_2\mI,
    \pi_1,\pi_2)$ of two multi-dimensional normal distributions $\mathcal{N}(\mu_1, \sigma_1\mI), \mathcal{N}(\mu_1, \sigma_1 \mI)$ with corresponding weights $\pi_1, \pi_2$:
    \begin{align*}
        &\mathrm{Var}_{x\sim \mathcal{D}}[x_i] = \pi_1  \sigma_1^2 + \pi_2 \sigma_2^2 + \pi_1 \mu_{1,i}^2 +  \pi_2 \mu_{2,i}^2 - ( \pi_1 \mu_{1,i} + \pi_2 \mu_{2,i} )^2
    \end{align*}
Specifically for the case of $\mu_1 = \mu = - \mu_2$, this equivalent with:
    \begin{align*}
        &\mathrm{Var}[x_i] = 1 + \mu_i^2  - ( \pi_1 \mu_{i} - \pi_2 \mu_{i} )^2
    \end{align*}
Hence:
    \begin{align*}
        &\mathbb{E} [x_i^2] = \mathrm{Var}[x_i] + ( \pi_1 \mu_{i} - \pi_2 \mu_{i} )^2 = 1 + \mu_i^2
    \end{align*}
    
So we can apply the following manipulations in the objective (The notation $[x_i^2]$ stands for the vector that has as entries the values of the vector $\vec{x}$ squared):

    \begin{align*}
    \displaystyle  \max_{ \vec{v}, \vec{w}} \min_{\vtheta, p} \ \ 
         \displaystyle  &\mathbb{E}_{ \vec{y} \sim real} \Big[ D_{\vec{v}} (\vec{y} ) +D_{\vec{w}}( \vec{y} ) \Big]
          -\\ 
          \displaystyle - &\mathbb{E}_{ \xi \sim \mathcal{N}(0,\mI), \zeta \sim \mathcal{N}(0,1)}
            \left[ G_p(\zeta)\cdot
                \Big( D_{\vec{v}}\big( - G_{\vtheta}( \vec{y})\big)+D_{\vec{w}}\big(G_{\vtheta}( \vec{y})\big) \Big)
            + \right. \\     
            & \indent + \big(1-G_p(\zeta)\big)\cdot \Big(D_{\vec{v}}\big(G_{\vtheta}( \vec{y}) \big) +D_{\vec{w}}\big(G_{\vtheta}( \vec{y}) \big) \Big)
            \Big] =     \end{align*}
    \begin{align*}
    \displaystyle  \max_{ \vec{v}, \vec{w}} \min_{\vtheta, p} \ \ 
         \displaystyle  &\mathbb{E}_{ \vec{y} \sim real} \Big[ \langle \vec{v}, \vec{y} \rangle +\langle \vec{w}, [y_i^2] \rangle \Big]
          - \\ - & \mathbb{E}_{ \xi \sim \mathcal{N}(0,\mI), \zeta \sim \mathcal{N}(0,1)}
            \left[ 
                (p + \zeta)\cdot
                \Big( \langle \vec{v}, - (\vec{z} + \vtheta ) \rangle +
                    \langle \vec{w}, [(\theta_i + \xi_i)^2] \rangle  \Big) +
                \right. \\ 
        & \indent + \left. \big(1- (p + \zeta) \big) \cdot \Big( \langle  \vec{v}, - (\vec{z} + \vtheta ) \rangle +
                    \langle \vec{w}, [\big( - (\theta_i + \xi_i) \big)^2] \rangle  \Big) \right]=     \end{align*}           \begin{align*}
        \displaystyle  \max_{ \vec{v}, \vec{w}} \min_{\vtheta, p} \ \ 
         \displaystyle  &\mathbb{E}_{ \vec{y} \sim real} \Big[ \langle \vec{v}, \vec{y} \rangle \Big] + \mathbb{E}_{ \vec{y} \sim real} \Big[\langle \vec{w}, [y_i^2] \rangle \Big]
          - \\ - & \mathbb{E}_{ \xi \sim \mathcal{N}(0,\mI), \zeta \sim \mathcal{N}(0,1)}
            \left[ 
                (p + \zeta)\cdot
                \Big( \langle \vec{v}, \vec{z} + \vtheta \rangle +
                    \langle \vec{w}, [(\theta_i + \xi_i)^2] \rangle  \Big) 
                + \right. \\ 
            & \indent + \left. \big(1- (p + \zeta) \big) \cdot \Big( \langle \vec{v}, -(\vec{z} + \vtheta ) \rangle +
                    \langle \vec{w}, [(\theta_i + \xi_i)^2] \rangle  \Big) \right] 
                        =\end{align*}
    \begin{align*}
    \displaystyle  \max_{ \vec{v}, \vec{w}} \min_{\vtheta, p} \ \ 
         \displaystyle  & \Big \langle \vec{v}, \mathbb{E}_{ \vec{y} \sim real} \Big[ \vec{y} \Big] \Big \rangle + \Big\langle  \vec{w}, \mathbb{E}_{ \vec{y} \sim real} \Big[ [y_i^2] \Big] \Big\rangle
          - \\ -& 
                p \cdot
                \Big( \Big\langle \vec{v}, 
                        \mathbb{E}_{ \xi \sim \mathcal{N}(0,\mI)}\Big[ (\vec{z} + \vtheta ) \Big] \Big\rangle +
                    \Big\langle \vec{w}, \mathbb{E}_{ \xi \sim \mathcal{N}(0,\mI)} \Big[ [(\theta_i + \xi_i)^2] \Big] \Big\rangle \Big) - \\
                - & \big(1- p \big) \Big( \cdot \Big\langle \vec{v},  \mathbb{E}_{ \xi \sim \mathcal{N}(0,\mI)}  \Big[  - (\vec{z} + \vtheta ) \Big] \Big\rangle +
                    \Big\langle \vec{w},  \mathbb{E}_{ \xi \sim \mathcal{N}(0,\mI)} \Big[ [(\theta_i + \xi_i)^2] \Big] \Big \rangle   \Big)= 
            \end{align*}
    \begin{align*}
    \displaystyle  \max_{ \vec{v}, \vec{w}} \min_{\vtheta, p} \ \ 
         \displaystyle  & \Big \langle \vec{v}, \Big( \pi_1 \vmu + \pi_2 (- \vmu) \Big) \Big\rangle + \Big\langle  \vec{w}, [1 + \mu_i^2] \Big\rangle
          - \\ -& 
                p \cdot
                \Big( \Big\langle \vec{v}, 
                      \vtheta \Big\rangle +
                    \Big\langle \vec{w}, [(\theta_i^2 + 1)] \Big] \Big\rangle \Big) 
                - \big(1- p \big) \cdot  \Big(  \Big\langle \vec{v},  - \vtheta \Big\rangle +
                    \Big\langle \vec{w},   [(\theta_i^2 + 1)] \Big \rangle   =
    \end{align*}
    \begin{align*}
            \displaystyle  &\max_{ \vec{v}, \vec{w}} \min_{\vtheta, p} \ \ 
    (\pi_1 - \pi_2) \vec{v}^T \vmu - 2p \vec{v}^T \vtheta + \vec{v}^T \vtheta + \sum_i w_i( \mu_i^2- \theta_i^2) 
    \end{align*}

}
\section{Proofs for \cref{sec:main}}
\subsection{Proof of \Cref{thm:hardness}}
We will reduce the problem of finding a Nash equilibrium in congestion games to the problem of finding a Nash equilibrium in {two-team} zero-sum games. The result then will follow since computing Nash equilibria in congestion games is $\CLS$-hard \citep{Aviad21}. 

{
As a recent result \citep{fearnley2021complexity} shows $\CLS$ is equal to the intersection of $\PLS$ and $\PPAD$, two important classes of total problems. $\PPAD$ captures diverse problems in combinatorics and (non-)cooperative game theory, like the $\varepsilon$-approximation of a mixed Nash Equilibrium in a graphical game or the computation of market equilibria. $\PLS$, for “Polynomial Local Search”, captures problems of finding a local minimum of an objective function f, in contexts where any candidate solution x has a local neighborhood within which we can readily check for the existence of some other point having a lower value of f. Many diverse local optimization problems have been shown complete for $\PLS$, attesting to its importance. Examples include searching for a local optimum of the TSP according to the Lin-Kernighan heuristic \citep{papadimitriou1992complexity} and finding pure Nash equilibria in many-player congestion games \citep{fabrikant2004complexity}. The complexity class $\CLS$ (“Continuous Local Search”) was introduced by Daskalakis and Papadimitriou \citep{daskalakis2011continuous} to classify various important problems that lie in both $\PPAD$ and $\PLS$. $\CLS$ is seen as a strong candidate for capturing the complexity of some of those important problems, like the general versions of Banach’s fixed point theorem, computation of KKT points, computation of gradient descent fixed points, etc.
}

In our reduction, a  \emph{congestion game}  is defined by the tuple $(N; E; (S_i)_{i \in N}; (c_e)_{e \in E})$ where $N$ is the set of \emph{agents},  $E$
is a set of \emph{resources} (also known as \emph{edges} or \emph{facilities}), and
each player $i$ has a set $S_i$ of subsets of $E$. Each strategy $s_i \in S_i$
is a set of edges (a \emph{path}), and $c_e$ is a cost (negative utility)
function associated with facility $e$. For a  strategy profile $\vec{s} = (s_1,s_2,\dots,s_N)$, the cost of player $i$ is given by $c_i(\vec{s}) = \sum_{e \in s_i} c_e(\ell_e(\vec{s}))$, where $\ell_e(\vec{s})$ is the
number of players using $e$ in $\vec{s}$ (the load of edge $e$). It is a well-known result \citep{rosenthal73} that congestion games exhibit a potential function $\Phi(\vec{s})$, that
\[
\Phi(\vec{s}) = \sum_{e\in E} \sum_{j=1}^{\ell_e(\vec{s})} c_e(j)
\]
with the property that if any agent $i$ changes her strategy to $s_i'$ it holds that
\[
\Phi(s_{i}',\vec{s}_{-i}) - \Phi(s_i, \vec{s}_{-i}) = c_i(s'_i,\vec{s}_{-i}) - c_i(s_i,\vec{s}_{-i}).
\]

\paragraph{{Succinct representation.}}
{
We assume that the congestion game and the team game are \emph{succinctly representable}~\citep{daskalakis2006game} and admits the \emph{polynomial expectation property}~\citep{daskalakis2006game}. That is, the expected cost of each player $i$, $c_i(\cdot)$, for any given mixed strategy profile $(s_1, \dots, s_n)$ is computed in time polynomial in:
\begin{inparaenum}[(i)]
 \item the number agents $n$,
 \item $\sum_i|S_i|$ where $S_i$ is the finite set of agent $i$'s strategies,
 \item the number of bits required to represent the mixed strategy profile $(s_1, \dots, s_n) $. 
\end{inparaenum}

}
 
\paragraph{Reduction.} Consider a congestion game $(N; E; (S_i)_{i \in N}; (c_e)_{e \in E})$ with $n = |N|$ players and potential function $\Phi.$ We define a team zero-sum game as follows:
Team $A$ has $n$ players, in which each agent $i$ chooses strategies from $S_i$. Team $B$ has $n$ players, with each agent $j$  having only one possible choice (singleton set of actions) call it $d$, i.e., these are dummy players. If players from Team $A$ choose strategy profile $\vec{s}$ (Team B has only one choice) then they get utility $u_{A}(\vec{s},d) = -\Phi(s)$. The utility members of Team B get is $u_{B}(\vec{s},d) = \Phi(s).$ 

Let $\vec{x}^*\defeq (x^*_{1},...,x^*_{n})$ and $(d,...,d)$ be a (possibly mixed) Nash equilibrium in the team zero-sum game we defined. We shall show that $(x^*_{1},...,x^*_{n})$ is a Nash equilibrium of the original congestion game and the reduction will be complete. Aiming for contradiction, suppose $(x^*_{1},...,x^*_{n})$ is not a Nash equilibrium of the original congestion game. Then there exists an agent $i$ that can deviate from strategy $x_i^*$ to $\tilde{x}_i$ and decrease her expected cost. Hence we have that
\begin{align*}
0&< \mathbb{E}_{\vec{s} \sim \vec{x}^*}[c_{i}(\vec{s})] - \mathbb{E}_{\vec{s} \sim (\tilde{x}^*_i,\vec{x}^*_{-i})}[c_{i}(\vec{s})]\\& = \mathbb{E}_{\vec{s} \sim \vec{x}^*}[\Phi(\vec{s})] - \mathbb{E}_{\vec{s} \sim (\tilde{x}^*_i,\vec{x}^*_{-i})}[\Phi(\vec{s})] \;\;\textrm{   (Property of potential)}.
\end{align*}
We conclude that $\mathbb{E}_{\vec{s}\sim \vec{x}^*}[u_{A} (\vec{s},d)] = -\mathbb{E}_{\vec{s}\sim \vec{x}^*}[\Phi(\vec{s})] < - \mathbb{E}_{\vec{s} \sim (\tilde{x}^*_i,\vec{x}^*_{-i})}[\Phi(\vec{s})] = \mathbb{E}_{\vec{s} \sim (\tilde{x}^*_i,\vec{x}^*_{-i})}[u_{A}(\vec{s},d)]$ which is a contradiction since $(x^*_1,...,x^*_n)$ is a Nash equilibrium for the team zero-sum game hence if player $i$ deviates, her payoff (i.e., the payoff of her Team) should not increase.

\subsection{Multiplayer Matching Pennies} \label{sec:tensorMMP}
Below, we present the exact definition of the (2x2) two-team of two-players matching pennies that we discuss in \Cref{sec:fom-fail}
    \begin{table}[h]
    \setlength{\extrarowheight}{4pt}
    \centering
    \begin{tabular}{cc|c|c|c|}
      & \multicolumn{1}{c}{} & \multicolumn{3}{c}{Team $B$}\\
      & \multicolumn{1}{c}{} & \multicolumn{1}{c}{$HH$}  & \multicolumn{1}{c}{$HT/TH$} & \multicolumn{1}{c}{$TT$} \\\cline{3-5}
      \multirow{3}*{Team $A$}  & $HH$ & $-1,1$ & $-1/2,1/2$ & $1,-1$ \\\cline{3-5}
      & $HT/TH$ & $1/2,-1/2$ & $0,0$ & $1/2,-1/2$ \\\cline{3-5}
      & $TT$ & $1,-1$ & $-1/2,1/2$ & $-1,1$ \\\cline{3-5}
    \end{tabular}
  \end{table}

\subsection{Proof of \Cref{thm:unstable}}
Since $(\vec{x}^*,\vec{y}^*)$ is not weakly-stable, there exist players $i,j$ from the same team (say $B$ without loss of generality) and strategies $k,l,l'$ so that if $i$ is forced to play $k$, then $\mat{J}$ 's best response is $l$ and that gives larger payoff than another strategy $l'$ in her support. Formally it holds that (by multi-linearity of $U$) 
\begin{equation}
\label{eq:refer}
(\textrm{payoff if }i,j \textrm{ choose }k,l) \;\; \frac{\partial ^2 U(\vec{x}^*,\vec{y}^*)}{\partial y_{ik} \partial y_{jl}} > \frac{\partial ^2 U(\vec{x}^*,\vec{y}^*)}{\partial y_{ik} \partial y_{jl'}}\;\; (\textrm{payoff if }i,j \textrm{ choose }k,l'), 
\end{equation}
and also $\frac{\partial U(\vec{x}^*,\vec{y}^*)}{\partial y_{jl}} = \frac{\partial U(\vec{x}^*,\vec{y}^*)}{\partial y_{jl'}}$ (*).
We shall show that $\nabla^2_{\vec{y}\vec{y}} U(\vec{x}^*,\vec{y}^*)$ has a strictly positive eigenvalue tangent in the product of simplices (note that if we were working with $A,$ we would show that $\nabla^2_{\vec{x}\vec{x}} U(\vec{x}^*,\vec{y}^*)$ has a strictly negative eigenvalue). Consider a vector of size $\sum_i |S_i|,$ (where $|S_i|$ is the cardinality of the strategy space of agent $i$ in Team $B$) which has $1$ at coordinates $(i,k), (j,l),$ $-1$ at coordinate $(j,l')$ and from which we subtract $y_i^*$; we denote by $\vec{v}$ the resulting vector. We shall show that $\vec{v}^{\top} \nabla^2_{\vec{x}\vec{x}} U(\vec{x}^*,\vec{y}^*) \vec{v} <0.$

By multilinearity of $U$ it follows that $\frac{\partial^2 U}{\partial y_{js}\partial y_{js'}} =0$ (**) for all agents $\mat{J}$  and strategies $s,s'$ and the same is true for $x$ variables (team A). We conclude that 
\begin{align*}
\frac{1}{2}\vec{v}^{\top} \nabla^2_{\vec{y}\vec{y}} U(\vec{x}^*,\vec{y}^*) \vec{v} &= \frac{\partial ^2U(\vec{x}^*,\vec{y}^*)}{\partial y_{ik} \partial y_{jl}} - \frac{\partial ^2 U(\vec{x}^*,\vec{y}^*)}{\partial y_{ik} \partial y_{jl'}} - \sum_{s\in S_i} y^*_{is}\left(\frac{\partial^2 U(\vec{x}^*,\vec{y}^*)}{\partial y_{is}\partial y_{jl}}-\frac{\partial^2 U(\vec{x}^*,\vec{y}^*)}{\partial y_{is}\partial y_{jl'}}\right)\\& = \frac{\partial ^2U(\vec{x}^*,\vec{y}^*)}{\partial y_{ik} \partial y_{jl}} - \frac{\partial ^2 U(\vec{x}^*,\vec{y}^*)}{\partial y_{ik} \partial y_{jl'}} - \left( \frac{\partial U(\vec{x}^*,\vec{y}^*)}{\partial y_{jl}}-\frac{\partial U(\vec{x}^*,\vec{y}^*)}{\partial y_{jl'}}\right) \\&\stackrel{(*)}{=}\frac{\partial ^2U(\vec{x}^*,\vec{y}^*)}{\partial y_{ik} \partial y_{jl}} - \frac{\partial ^2 U(\vec{x}^*,\vec{y}^*)}{\partial y_{ik} \partial y_{jl'}} \stackrel{\Cref{eq:refer}}{>}0.
\end{align*}
Therefore $\nabla^2_{\vec{y}\vec{y}} U(\vec{x}^*,\vec{y}^*)$ has a positive eigenvalue and as a result
\begin{equation}\label{eq:before}
\mat{R}:=\left(\begin{array}{cc}
- \nabla^2_{\vec{x}\vec{x}} U(\vec{x}^*,\vec{y}^*) & \vec{0} 
\\ \vec{0} &   \nabla^2_{\vec{y}\vec{y}} U(\vec{x}^*,\vec{y}^*)
\end{array}\right)
\end{equation}
must have a positive and a negative eigenvalue (since the trace is zero).

We consider the Jacobian of the GDA dynamics at $(\vec{x}^*,\vec{y}^*)$. The corresponding matrix is the following:
\begin{equation}\label{eq:Jh}
\jgda = \mat{I}+\eta \left(\begin{array}{cc}
- \nabla^2_{\vec{x}\vec{x}} U(\vec{x}^*,\vec{y}^*) & - \nabla^2_{\vec{x}\vec{y}} U(\vec{x}^*,\vec{y}^*)
\\ \nabla^2_{\vec{y}\vec{x}} U(\vec{x}^*,\vec{y}^*) &   \nabla^2_{\vec{y}\vec{y}} U(\vec{x}^*,\vec{y}^*)
\end{array}\right),
\end{equation}
We will show that $\jgda$ has an eigenvalue (possible complex) with {an} absolute value greater than one. It suffices to show that $\jgda- \mat{I}$ has an eigenvalue with a positive real part (because then $\jgda$ would have an eigenvalue with a real part greater than 1 and hence magnitude greater than one). Due to (**), we get that $\jgda- \mat{I}$ has trace zero. To reach contradiction suppose that no eigenvalue of $\jgda- \mat{I}$ has a positive real part, then all eigenvalues of $\jgda- \mat{I}$ should be imaginary or zero. But an imaginary eigenvalue (that is not zero) also results in an eigenvalue with a magnitude greater than one for $\jgda$, therefore all eigenvalues of $\jgda- \mat{I}$ should be zero. We use Ky Fan inequalities which states that the sequence (in decreasing order) of the eigenvalues of $\frac{1}{2}(\mat{H}+\mat{H}^{\top})$ majorizes the real part of the sequence of the eigenvalues of $\mat{H}$ (see \citep[p. 4]{kyfan}) for any matrix $\mat{H}$.
We choose $\mat{H} = \frac{1}{\eta} \cdot\left( \jgda - \mat{I}\right)$. Since $\mat{R}$ has both a negative and a positive eigenvalue, we get that $\mat{H}$ has an eigenvalue with a negative real part. The claim follows since $\mat{H}$ has a trace equal to zero, thus it must have an eigenvalue with a positive real part as well.  

We conclude that $\jgda$ has an eigenvalue with {an} absolute value greater than one. Using \citep[Theorem 2.2]{daskalakis2018limit}, it occurs that the set of initial conditions so that GDA converge to $(\vec{x}^*,\vec{y}^*)$ is of measure zero (for the particular choice of the {step size}).

\subsection{Proof of \Cref{lemma:onemixed}}
First, we start with the min-max optimization objective of the GMP game with $\omega = \frac{1}{2}$:
\begin{equation*}
\begin{split}
\min_{\vec{x} \in [0,1]^2} &\max_{\vec{y} \in [0,1]^2} -x_1 x_2 y_1y_2 -(1-x_1)(1-x_2)(1-y_1)(1-y_2)  +   x_1x_2(1-y_1)(1-y_2) +\\& +(1-x_1)(1-x_2)y_1y_2 + \omega \left(1-x_1x_2-(1-x_1)(1-x_2)\right)\left(y_1y_2+(1-y_1)(1-y_2)\right)-\\&  -\omega 
\left(1-y_1y_2-(1-y_1)(1-y_2)\right)\left(x_1x_2+(1-x_1)(1-x_2)\right)
\end{split}
\end{equation*}
or after simplification{,} it is equivalent with
\begin{equation}
\label{eq:minmax}
\min_{\vec{x} \in [0,1]^2} \max_{\vec{y} \in [0,1]^2} (\omega+1) (x_1+x_2)+(1-\omega)(y_1+y_2) - (x_1+x_2)(y_1+y_2)-2\omega x_1x_2+2\omega y_1y_2.
\end{equation}
\begin{remark}
We highlight that the min-max optimization objective is \textit{multilinear} and the degree of each variable in every summand is at most one (total degree is $2$). Moreover, we note that due to the nonconvexity-nonconcavity of the function above, the max-min is not equal to the min-max.
\end{remark}

Let $(x_1^*,x_2^*,y_1^*,y^*_2)$ be a Nash equilibrium. Assuming $x_1^*,x_2^*,y_1^*,y_2^* \in (0,1)$ from \Cref{eq:minmax} and first-order conditions we get the system of equations
\begin{enumerate}
\item $\omega+1 - y_1^*-y_2^*-2\omega x^*_2 = 0,$ 
\item $\omega+1 - y_1^*- y_2^*-2\omega x^*_1 = 0,$
\item $1-\omega - x_1^*-x_2^* +2\omega y^*_2 = 0,$
\item $1-\omega - x_1^* - x_2^*+2\omega y^*_1 = 0.$
\end{enumerate}
Combining the first two equations, we have $x_1^* = x_2^*$ and combining the last two it follows $y_1^* = y_2^*.$ Dividing equation one by $\omega$ and subtracting three, it follows that $ \frac{1}{\omega} - 2\frac{y_1^*}{\omega} + \omega - 2\omega y_1^*=0.$ Hence we conclude that $y_1^* = \frac{1}{2}$. As a result $y_2^* = \frac{1}{2}$ and substituting in first equation $x_1^* = x_2^* = \frac{1}{2}.$ 

\begin{itemize}
\item Assume now that $x_1^*=0$ and $x_2^*, y_1^*, y_2^* \in (0,1).$ Following the same idea, now only equations 2, 3, 4 hold and instead of the first we have the constraint $\omega+1 - y_1^*-y_2^*-2\omega x^*_2 \geq 0.$ From 3, 4 we conclude that $y_1^*=y_2^*$ and using 2 it holds that $y_1^* = y_2^* = \frac{1+\omega}{2}.$ Using 3 follows that $1-\omega -x_2^* + \omega(\omega+1)=0.$ Thus $x_2^* = 1 + \omega^2>1$ (this is not possible because $x_2^* \in [0,1]$). 

\item Consider the case that $x_1^*=1$ and $x_2^*, y_1^*, y_2^* \in (0,1).$ Only equations 2, 3, 4 hold and instead of the first we have the constraint $\omega+1 - y_1^*-y_2^*-2\omega x^*_2 \leq 0.$ From 3, 4 we conclude that $y_1^*=y_2^*$ and using 2 it holds that $y_1^* = y_2^* = \frac{1-\omega}{2}.$ Using 3 follows that $-\omega -x_2^* + \omega(1-\omega)=0.$ Thus $x_2^* =  -\omega^2<0$ (this is not possible because $x_2^* \in [0,1]$). By symmetry the same happens when $x_2^*=0$ or $x_2^*=1$ and $x_1^*,y_1^*,y_2^* \in (0,1).$

\item Case $x_1^*=x_2^* =0$ and $y_1^*,y_2^* \in (0,1)$. Using 3, 4 we get $y_1^*=y_2^* = \frac{\omega-1}{2\omega}<0$ (this is not possible). 
\item Case $x_1^*=x_2^* =1$ and $y_1^*,y_2^* \in (0,1)$. Using 3, 4 we get $y_1^*=y_2^* = \frac{\omega+1}{2\omega}>1$ (this is not possible).
\item Case $x_1^*=0$ and $x_2^* =1$ and $y_1^*,y_2^* \in (0,1)$. Using 3, 4 we get $y_1^*=y_2^* = \frac{1}{2}$. Moreover one becomes $\omega-2\omega x^*_2 \geq 0$ and two $\omega-2\omega x^*_1 \leq 0$, that is $x_1^* \geq \frac{1}{2}$ and $x_2^* \leq \frac{1}{2}$ (contradiction). The case $x_1^*=1$ and $x_2^* =0$ and $y_1^*,y_2^* \in (0,1)$ is symmetric. 
\end{itemize}
Similarly one can consider the case where the $\vec{y}$ team plays pure and $x^*_1,x^*_2 \in (0,1).$  One can also check that all possible pure strategy profiles are not Nash equilibria.

\subsection{Proof of \Cref{thm:failure}}
We split the proof into 3 parts. Before we start the proof, note that the Hessian of $U$ \Cref{eq:hessta} has infinity norm less than 4 (since $\omega \in (0,1)$), so $U$ has gradient Lipschitz with $L \leq 4.$ Thus for the rest of the proof for GDA, we choose $\eta_{\textrm{OGDA}} < \frac{1}{4}$. \\\\
\textbf{GDA.} For GDA the proof will be straightforward. 
We will show that $(x_1^*,x_2^*,y_1^*,y^*_2)$ is a weakly Nash equilibrium. Then the claim about GDA will follow because of \Cref{lemma:onemixed}, \Cref{thm:unstable} and Remark \ref{rem:fixedpointsNash}.

Assume that player $x_1$ fixes his strategy to $x_1 =0.$ and $y_1,y_2$ keep their strategy $(\frac{1}{2},\frac{1}{2})$. We shall show that $x_2$ is not indifferent in his support and would like to change his mixed strategy $x_2 = \frac{1}{2}$ to pure. When $x_1=0$ and $y_1 = y_2 =\frac{1}{2}$ the payoff of Team $A$ ($x$ variables) becomes $-\omega x_2 -1 +\frac{\omega}{2}.$ Since $\omega \in (0,1),$ $x_2$ prefers to play $x_2=0$ (instead of $\frac{1}{2}$ she had). We conclude that $(\frac{1}{2},\frac{1}{2},\frac{1}{2},\frac{1}{2})$ is not a weakly-stable Nash equilibrium.\\\\ 
\textbf{OGDA.} The Jacobian of the  update rule of OGDA dynamics (use the same machinery of Section 3 in \citep{daskalakis2018limit}) is the following:\\\\
\begin{equation}\label{eq:Jg}
\mat{J}_{\mathrm{OGDA}} = \left(\begin{array}{cccc}
\mat{I} - 2\eta_{\mathrm{OGDA}}\nabla^2_{\vec{x}\vec{x}} U & - 2\eta_{\mathrm{OGDA}}\nabla^2_{\vec{x}\vec{y}} U & \eta_{\mathrm{OGDA}}\nabla^2_{\vec{x}\vec{x}} U & \eta_{\mathrm{OGDA}}\nabla^2_{\vec{x}\vec{y}} U
\\ 2\eta_{\mathrm{OGDA}}\nabla^2_{\vec{y}\vec{x}} U & \mat{I} + 2\eta_{\mathrm{OGDA}}\nabla^2_{\vec{y}\vec{y}} U & -\eta_{\mathrm{OGDA}}\nabla^2_{\vec{y}\vec{x}} U & -\eta_{\mathrm{OGDA}}\nabla^2_{\vec{y}\vec{y}} U
\\ \mat{I} & \vec{0} & \vec{0}& \vec{0}
\\ \vec{0} & \mat{I} & \vec{0} & \vec{0}
\end{array}\right),
\end{equation}
where $U$ is the payoff of Team $B$ (max) and $-U$ is the payoff of team $A$. Substituting for MPG payoff at Nash equilibrium, we have
\[
\nabla^2_{\vec{x}\vec{x}} U = \left(\begin{array}{cc}
0 & -2\omega\\
-2\omega & 0
\end{array}\right), \nabla^2_{\vec{y}\vec{y}} U = \left(\begin{array}{cc}
0 & 2\omega\\
2\omega& 0
\end{array}\right), \nabla^2_{\vec{x}\vec{y}} U = \left(\begin{array}{cc}
-1 & -1\\
-1 & -1
\end{array}\right). 
\] 
The corresponding Jacobian matrix becomes:
\begin{equation}\label{eq:Jall}
\mat{J}_{\mathrm{OGDA}} = \left(\begin{array}{cccc}
\mat{I} & \vec{0} & \vec{0}& \vec{0}
\\ \vec{0} & \mat{I} & \vec{0} & \vec{0}
\\ \mat{I} & \vec{0} & \vec{0}& \vec{0}
\\ \vec{0} & \mat{I} & \vec{0} & \vec{0}
\end{array}\right)+\eta_{\mathrm{OGDA}}\left(\begin{array}{cccccccc}
0 &4\omega &2 &2&0 &-2\omega&-1&-1\\
4\omega &0 &2 &2&-2\omega &0&-1&-1\\
-2 &-2 &0 &4\omega &1 &1 &0 &-2\omega\\
-2 &-2 &4\omega &0 &1 &1 &-2\omega &0\\
0&0&0&0&0&0&0&0\\
0&0&0&0&0&0&0&0\\
0&0&0&0&0&0&0&0\\
0&0&0&0&0&0&0&0
\end{array}\right),
\end{equation}

It turns out that matrix \Cref{eq:Jall} has a characteristic polynomial (computed with the help of Mathematica) 
\begin{equation}\label{eq:polynomial}
\begin{split}
\pi(\lambda) =  (4 (1 + \omega^2) \eta_{\mathrm{OGDA}}^2 (1 - 2 \lambda)^2 + (\lambda-1)^2 \lambda^2 - 4 \omega \eta_{\mathrm{OGDA}} \lambda (1 - 3 \lambda + 2 \lambda^2))\\ \cdot 
((\lambda -1 ) \lambda + 
   2 \omega \eta_{\mathrm{OGDA}} (2\lambda-1))^2.
   \end{split}
   \end{equation}

A root of $\pi(\lambda)$ is $\frac{1}{2} \left(1 + 4 \eta_{\mathrm{OGDA}} j +4\omega  \eta_{\mathrm{OGDA}} +\sqrt{1-16\eta^2_{\mathrm{OGDA}}+32j\omega \eta^2_{\mathrm{OGDA}}+16\omega^2\eta^2_{\mathrm{OGDA}}}\right) ,$
which is in absolute value greater than one for $0<\eta_{\mathrm{OGDA}} \leq \omega$.
It is also easy to see the Hessian of $U$, that is:
\begin{equation}\label{eq:hessta}
\left(\begin{array}{cccc}
\nabla^2_{\vec{x}\vec{x}} U &\nabla^2_{\vec{x}\vec{y}} U\\
\nabla^2_{\vec{y}\vec{x}} U & \nabla^2_{\vec{y}\vec{y}} U
\end{array}\right)
=
\left(\begin{array}{cccc}
0 &-2\omega &-1 &-1\\
-2\omega &0 &-1 &-1\\
-1 &-1 &0 &2\omega \\
-1 &-1 &2\omega &0 
\end{array}\right)
\end{equation}
is invertible for $\omega\in (0,1)$ (this is true as the eigenvalues are $-2\omega,2\omega,-2\sqrt{1+\omega^2},2\sqrt{1+\omega^2},$ non of which is zero). From Theorem 3.2 in \citep{daskalakis2018limit} follows that the initial conditions so that OGDA converges to the Nash equilibrium is of measure zero for step size $\eta_{\mathrm{OGDA}} <\frac{1}{2L} \leq \frac{1}{8}$ (where $L$ is the Lipschitz constant of $\nabla U$). We choose $\eta_{\mathrm{OGDA}} \leq \min (\frac{1}{8},\omega)$ and the claim follows.\\\\
\textbf{EG.} The Jacobian of EG dynamics computed at the Nash equilibrium (fixed point) is given below (made use of chain rule):
\begin{align}\label{eq:JEG}
\jeg = \mat{I} &+\eta_{\textrm{EG}} \left(\begin{array}{cc}
- \nabla^2_{\vec{x}\vec{x}} U(\vec{x}^*,\vec{y}^*) & - \nabla^2_{\vec{x}\vec{y}} U(\vec{x}^*,\vec{y}^*) \nonumber
\\ \nabla^2_{\vec{y}\vec{x}} U(\vec{x}^*,\vec{y}^*) &   \nabla^2_{\vec{y}\vec{y}} U(\vec{x}^*,\vec{y}^*)
\end{array}\right) \\&+ \eta_{\textrm{EG}}^2\left(\begin{array}{cc}
- \nabla^2_{\vec{x}\vec{x}} U(\vec{x}^*,\vec{y}^*) & - \nabla^2_{\vec{x}\vec{y}} U(\vec{x}^*,\vec{y}^*)
\\ \nabla^2_{\vec{y}\vec{x}} U(\vec{x}^*,\vec{y}^*) &   \nabla^2_{\vec{y}\vec{y}} U(\vec{x}^*,\vec{y}^*)
\end{array}\right)^2.
\end{align}
We substitute with the values and we get
\begin{align}
\jeg = \mat{I} &+\eta_{\textrm{EG}} \left(\begin{array}{cccc}
0&2\omega & 1&1
\\ 2\omega&0&1&1\\
-1&-1&0&2\omega\\
-1& -1& 2\omega&0
\end{array}\right) \nonumber
\\&+ \eta_{\textrm{EG}}^2\left(\begin{array}{cccc}
4\omega^2-2 & -2&4\omega&4\omega\\
-2& 4\omega^2-2&4\omega&4\omega\\
-4\omega&-4\omega&-2+4\omega^2&-2\\
-4\omega&-4\omega&-2&-2+4\omega^2
\end{array}\right). 
\end{align}
The eigenvalues of $\jeg$ are $ 1+\eta_{\textrm{EG}}(-2\omega + 4 \eta_{\textrm{EG}} \omega^2),$ $1+\eta_{\textrm{EG}}(-2\omega + 4 \eta_{\textrm{EG}} \omega^2),$\\ $1+\eta_{\textrm{EG}} \left(-4 \eta_{\textrm{EG}} + 2\omega + 4 \eta_{\textrm{EG}} \omega^2 - 
  2\sqrt{-1 - 8 \eta_{\textrm{EG}} \omega - 16 \eta_{\textrm{EG}}^2 \omega^2}\right),$ and \\$ 1+\eta_{\textrm{EG}}  \left(-4 \eta_{\textrm{EG}} + 2\omega + 4 \eta_{\textrm{EG}} \omega^2 + 
  2\sqrt{-1 - 8 \eta_{\textrm{EG}} \omega - 16 \eta_{\textrm{EG}}^2 \omega^2}\right).$
If $-4\eta_{\textrm{EG}} + 2\omega + 4 \eta_{\textrm{EG}} \omega^2>0$ then $\jeg$ will have an eigenvalue with absolute value greater than one. A sufficient condition for that is $\eta_{\textrm{EG}} \leq \frac{\omega}{2}.$ Finally, for the same choice of stepsizes we have that $\jeg$ is invertible for all $(\vec{x},\vec{y})$, hence the EG dynamics is a local diffeomorphism. From Theorem 2 in arxiv version of \citep{LeePPSJR19} the claim for EG method follows.

{
\textbf{OMWU.} The Jacobian of the  update rule of OMWU dynamics (use the same machinery of Section 3 in \citep{daskalakis2019last}) is the following:\\\\
\begin{equation}\label{eq:Jomw}
\jomwu = \left(\begin{array}{cccccccc}
1 & \eta \omega &  \eta & \eta & 0& -\frac{\eta \omega}{2}  &-\frac{\eta }{2} &-\frac{\eta}{2}
\\ \eta\omega & 1 & \eta & \eta &-\frac{\eta \omega}{2}&0&-\frac{\eta}{2}& -\frac{\eta}{2}
\\ -\eta & -\eta & 1&  \eta\omega & \frac{\eta}{2} & \frac{\eta}{2} &0 &-\frac{\eta \omega}{2}
\\ -\eta & -\eta &\eta \omega&1&\frac{\eta}{2}&\frac{\eta}{2}&-\frac{\eta\omega}{2}&0
\\ 1 & 0&0&0&0&0&0&0
\\ 0 & 1&0&0&0&0&0&0
\\0 & 0&1&0&0&0&0&0
\\0 & 0&0&1&0&0&0&0
\end{array}\right).
\end{equation}
It turns out that matrix \Cref{eq:Jomw} has a characteristic polynomial (retrieved with the help of Mathematica) 
\begin{equation}
\begin{split}
\pi(\lambda) = \frac{1}{16} \left((4 + \omega^2) \eta^2 (1 - 2 \lambda)^2 + 
   4 (\lambda-1)^2 \lambda^2 - 4 \omega \eta \lambda (1 - 3 \lambda + 2 \lambda^2)\right)\cdot\\\left(2 (\lambda-1) \lambda + \omega \eta ( 2 \lambda-1)\right)^2.
   \end{split}
   \end{equation}
One root of the polynomial $\pi(\lambda)$ above is 
$\frac{1}{2} (1 + 2\eta j  + \eta\omega  + \sqrt{1-4\eta^2+4\eta^2\omega j+\eta^2\omega^2}).$ As in the analysis of OGDA, it turns out that for $0<\eta < \omega$ the aforementioned root has an absolute value greater than one. 

Therefore we conclude that the Nash equilibrium where each agent plays $(\frac{1}{2}, \frac{1}{2})$ is repelling (and thus the fact that the set of initial conditions so that OMWU converges to that particular fixed point is of measure zero is derived from standard arguments from \citep{LeePPSJR19}. To conclude the proof we need to exclude that OMWU will stabilize in other fixed points. Note that OMWU can stabilize to points that are not Nash equilibria (e.g., all pure strategy profiles are fixed points). To exclude such stabilization, the trick is that if OMWU stabilizes to a point, it should be (approximate coarse correlated equilibrium where the approximation depends on the size of the stepsize). However, none of the pure strategy profiles is a $\omega/2$- approximate coarse correlated equilibrium, so OMWU does not stabilize. Hence if we choose $0<\eta < \omega/2,$ OMWU does not stabilize. 
}

\subsection{Proof of \Cref{thm:success}}

    As previously, let a zero-sum game with payoff function $U(\vec{x}, \vec{y})$ for the maximizing team (it can either be a single player or more players in each team). The maximizing team controls $\vec{y}$ while the minimizing team controls $\vec{x}$.
    
    It is easy to retrieve a black-box argument for the proof of \Cref{thm:success} using the proof \citep{hassouneh2004washout}. For the sake of completeness, we reprove the theorem from the bottom up in our own vocabulary. \citet{hassouneh2004washout} extend the proof \citep{bazanella2000} to the discrete-time case without using the singular-perturbation analysis that the latter use.

    Let $\mat{A} = \jgda = \mat{I} + \eta \left(\begin{array}{cc} -\nabla^2_{\vec{x}\vec{x}}U(\vec{x}^*,\vec{y}^*) &-\nabla^2_{\vec{y}\vec{x}}U(\vec{x}^*,\vec{y}^*)\\\nabla^2_{\vec{x}\vec{y}}U(\vec{x}^*,\vec{y}^*) &\nabla^2_{\vec{y}\vec{y}}U(\vec{x}^*,\vec{y}^*)  \end{array}\right)$ be the Jacobian of GDA around a Nash equilibrium. and $\vec{z}$ be the concatenation of $\vec{x}$ and $\vec{y}$. We can locally represent the dynamic of GDA as a linear time-invariant dynamical system:
    \begin{align}
        \vec{z}^{(t+1)} &= \mat{A} \vec{z}\step{t} + \mat{B} \vec{u}\step{t} \\
        \vec{u}^{(t+1)} &= \vec{0}
    \end{align}
    
    Here, the feedback law, $\vec{u}$, is trivially zero, but we note it for the sake of clarity.
    
    By assumption, $\mat{A} = \mat{J}_{\mathrm{GDA}} $ is invertible, hence full-rank. Also, $\mat{B}$ has to be the identity matrix, $\mat{B}=\mat{I}$, as we want every player to be able to make decisions with access to their own estimation vector.
    
    
    From a known theorem (the PBH controllability test) the pair of matrices $(\mat{A},\mat{B})$, with $\mat{A} \in \R^{n\times n}$ is controllable if and only if, for every eigenvalue $\lambda$ of $\mat{A}s$, the rank of the block matrix $[\mat{A} - \lambda \mat{I}, \mat{B}]$  is equal to $n$. As we know, stabilizability implies controllability. Hence, the pair of matrices $(\mat{A},\mat{B}) = (\mat{J}_{\mathrm{GDA}}, \mat{I})$ is stabilizable.

    We can now write down the dynamic of \Cref{eq:KPV} localized around a given (possibly Lyapunov unstable) Nash equilibrium as the following linear time-invariant system:
    
    \begin{gather}
    \left\{
    \begin{array}{lcl}
        \vec{z}^{(t+1)} &=& \underbrace{\jgda}_{\mat{A} } \vec{z}\step{t} + \underbrace{\mat{I} \cdot \mat{K}}_{B\cdot \mat{K}} \vec{u}\step{t} \\
        \boldsymbol{\theta}^{(t+1)} &=& \boldsymbol{\theta}\step{t} + \mat{P} \big( \vec{z}\step{t} - \boldsymbol{\theta}\step{t}
        \big)
        \\
        \vec{u}^{(t+1)} &=& \vec{z}\step{t} - \boldsymbol{\theta}\step{t}
    \end{array}
    \right.
    .
    \end{gather}

    This dynamic can be represented in a block matrix fashion as:
    \begin{equation}
       \mat{M} := \left(\begin{array}{cc}
         \jgda + \mat{K} & - \mat{K}  \\
         \mat{P} & \mat{I} - \mat{P}
        \end{array}
        \right).
    \end{equation}

    Knowing that eigenvalues are preserved under similarity transformations, we define two similarity transformations $\mat{T}_1, \mat{T}_2$ which we will apply consecutively on $\mat{M}$.
    
    Let $\mat{Q}$ be any matrix. We define:
    \begin{gather}
        \mat{T}_1 = 
        \begin{pmatrix}
            \mat{I} & \mat{0} \\
            \mat{0} & \mat{P}^{-1}
        \end{pmatrix}, 
        \quad
        \mat{T}_1^{-1} = 
        \begin{pmatrix}
            \mat{I} & \mat{0} \\
            \mat{0} & \mat{P}
        \end{pmatrix}
    \end{gather}
    
    and:
    \begin{gather}
        \mat{T}_2 = 
        \begin{pmatrix}
            \mat{I} & \mat{Q} \\
            \mat{0} & \mat{I}
        \end{pmatrix}, 
        \quad
        \mat{T}_2^{-1} = 
        \begin{pmatrix}
            \mat{I} & -\mat{Q} \\
            \mat{0} & \mat{I}
        \end{pmatrix}
    \end{gather}
    
    By consecutively applying the latter transformations we get:
    \begin{align}
        \mat{M}' =& \mat{T}_2 \mat{T}_1 \mat{M} \mat{T}_1^{-1}\mat{T}_2^{-1} = \nonumber \\
        =&
        \left(
        \begin{array}{c|c}
            \jgda + \mat{K} + \mat{Q} & -\jgda \mat{Q} - \mat{K}\mat{Q} - \mat{Q}^2 - \mat{K}\mat{P} + \mat{Q} - \mat{Q}\mat{P} \\\hline
             \mat{I} & -\mat{Q} + \mat{I} -\mat{P}
        \end{array}
        \right)
    \end{align}
  
    We now want to eliminate the top-left entry of $\mat{M}'$ in order to have an easy application of the Schur complement equations. We are free to do so if we can find an appropriate matrix $\mat{Q}$. Remember that $\mat{Q}$ can be any matrix.
      
    First, we write matrices $\mat{P}$ and $\mat{Q}$ as (non-singular) perturbations of some appropriate matrices for $\epsilon>0$ sufficiently small: 
    \begin{align}
        \mat{P} &= \epsilon \mat{P}_1 \\
        \mat{Q} &= \mat{Q}_0 +\epsilon \mat{Q}_1 + O (\epsilon^2)
    \end{align}

    Setting the top-right entry of $\mat{M}'$ equal to $\mat{0}$:
    \begin{align}
        -\jgda \mat{Q} - \mat{K}\mat{Q} - \mat{Q}^2 - \mat{K}\mat{P} + \mat{Q} - \mat{Q}\mat{P} &= 0        \label{eq:settozero}
 \Rightarrow \\
        -\jgda \Big( \mat{Q}_0 +\epsilon \mat{Q}_1 + O (\epsilon^2) \Big) - \mat{K} \Big( \mat{Q}_0 +\epsilon \mat{Q}_1 + O (\epsilon^2) \Big) - \Big( \mat{Q}_0 +\epsilon \mat{Q}_1 + O (\epsilon^2) \Big)^2 - &\nonumber \\ - \mat{K}\epsilon \mat{P}_1 + \Big( \mat{Q}_0 +\epsilon \mat{Q}_1 + O (\epsilon^2) \Big) - \Big( \mat{Q}_0 +\epsilon \mat{Q}_1 + O (\epsilon^2) \Big)\epsilon \mat{P}_1 &= 0
    \end{align}

    Gathering terms w.r.t. to the power of $\epsilon$ we get, for $O(1)$ terms:
    \begin{gather}
        (\jgda + \mat{K} - \mat{I} + \mat{Q}_0)\mat{Q}_0 = \mat{0}.
    \label{eq:oone}
    \end{gather}
    
    And for $O(\epsilon)$ terms we get:
    \begin{gather}
        \mat{Q}_1( \jgda + \mat{K} - \mat{I} ) + (\jgda - \mat{I} ) \mat{P}_1 = \mat{0}.
    \label{eq:oeps}\end{gather}.

    From \eqref{eq:oone} we get that 
    
    \begin{equation}
        \mat{Q}_0 = \mat{0}  \quad \text{or} \quad \mat{Q}_0 = - \jgda - \mat{K} + \mat{I}.
    \end{equation}
    
    From \eqref{eq:oeps}, by observing that we can always select a $\mat{K}$ such that $1$ is not an eigenvalue of $\jgda + \mat{K}$ we get that:
    
    \begin{equation}
        \mat{Q}_1 = - (\jgda - \mat{I})\mat{P}_1 ( \jgda + \mat{K} - \mat{I} ) ^{-1}.
    \end{equation}

    We observe that $\mat{Q}_1$ is uniquely defined through a matrix inversion. Hence, from the Implicit Function Theorem applied to \eqref{eq:settozero} we get that there exists a locally unique matrix $\mat{Q}$:
    \begin{equation}
        \mat{Q} = - \jgda - \mat{K} + \mat{I} - \epsilon ( \jgda - \mat{I} ) \mat{P}_1 ( \jgda + \mat{K} - \mat{I})^{-1} + O(\epsilon^2)
    \end{equation}

    We return to matrix $\mat{M}'$ and substitute $\mat{Q}$ with what we calculated in order to get:
    
    \begin{gather}
        \mat{M}' = \left( \begin{array}{c|c}
            \mat{I} - \epsilon ( \jgda - \mat{I} ) \mat{P}_1 ( \jgda + \mat{K} - \mat{I})^{-1} + O(\epsilon^2) & \mat{0}  \\ \hline
            \mat{I} & \mat{\Lambda}
        \end{array}
        \right)
    \end{gather}
    
    With $\mat{\Lambda} =  \jgda + \mat{K} + \epsilon\Big( (\jgda - \mat{I} )\mat{P}_1(\jgda + \mat{K} - \mat{I})^{-1} - \mat{P}_1  \Big) + O(\epsilon^2)$.
    
    We have assumed 
    $\eta \left( \begin{array}{cc}
-\nabla^{2}_{\vec{x}\vec{x}}U(\vec{x}^*,\vec{y}^*) & -\nabla^{ 2}_{\vec{x}\vec{y}}U(\vec{x}^*,\vec{y}^*)\\
\nabla^{ 2}_{\vec{y}\vec{x}}U(\vec{x}^*,\vec{y}^*) &\nabla^{ 2}_{\vec{y}\vec{y}}U(\vec{x}^*,\vec{y}^*)
\end{array}\right) = \jgda - \mat{I} $ to be invertible, hence in order for $\mat{M}'$ to be stable it suffices that:
\begin{itemize}
    \item $\mat{I} - \epsilon(\jgda - \mat{I}) \mat{P}_1 (\jgda + \mat{K} - \mat{I}) ^{-1} + O(\epsilon^2)$ is stable
    \item $\jgda + \mat{K}$ is stable with eigenvalues away from $1$,\\ so that term $ \epsilon\Big( (\jgda - \mat{I} )\mat{P}_1(\jgda + \mat{K} - \mat{I})^{-1} - \mat{P}_1  \Big) $ does not cause the eigenvalues of $\jgda + \mat{K}$ to get a a magnitude greater than $1$ (not to become unstable).
\end{itemize} 
    For the first item it easy to see that an appropriate $\mat{P}_1$ always exists, that is, $\mat{P}_1 = (\jgda - \mat{I})^{-1} ( \jgda+ \mat{K} - \mat{I})$. 
    
    As for the second item, it follows from the fact that $(\jgda, \mat{I})$ is stabilizable as we have demonstrated.




{
\subsection{Proof of \Cref{thm:sufficient} }
We first compute the Jacobian of KPV-GDA dynamics \Cref{eq:KPV} where we have eliminated one variable per agent (to avoid using the projection operator). The Jacobian has the following form:
\begin{equation}\label{eq:JKPV}
\mat{J}_{\textrm{KPV-GDA}} = \mat{I} +\eta\left(\begin{array}{cc}
 k\cdot \mat{I} +  \left(\begin{array}{cc} -\nabla^2_{\vec{x}\vec{x}}U(\vec{x}^*,\vec{y}^*) &-\nabla^2_{\vec{y}\vec{x}}U(\vec{x}^*,\vec{y}^*)\\\nabla^2_{\vec{x}\vec{y}}U(\vec{x}^*,\vec{y}^*) &\nabla^2_{\vec{y}\vec{y}}U(\vec{x}^*,\vec{y}^*)  \end{array}\right) & -  k\cdot \mat{I}
\\  p\cdot \mat{I} &  - p\cdot \mat{I}
\end{array}\right).
\end{equation}
It suffices to show that the spectral radius of $\mat{J}_{\textrm{KPV-GDA}}$ is less than one for step size $\eta$ small enough. 
We first show the claim below
\begin{claim}\label{cl:help}
Let $\mat{J} = \mat{I} + \eta \mat{M}$ where $\mat{I}$ is the identity matrix. Suppose that $\mat{M}$ has all of its eigenvalues with a real part that is negative. Then, there exists an interval $(0,\eta_0)$ so that if $\eta \in (0,\eta_0)$ then $\mat{J}$  has all of its eigenvalues with an absolute value less than one.
\end{claim}
\begin{proof}
Let $a_i+ b_i \cdot j$ be an eigenvalue of $\mat{M}$ such that $a_i< 0.$ Assume that $\eta < -\frac{2a_i}{a_i^2+b_i^2}$, then the corresponding eigenvalue of $\mat{J}$  is $1 + \eta (a_i+b_i\cdot j)$, the magnitude of which is $(1+\eta a_i)^2 + \eta^2 b_i^2 = 1 + 2\eta a_i + \eta^2 (a_i^2+b_i^2) <1$ since $-2a_i > \eta (a_i^2+b_i^2)$ (by assumption). Hence we can choose $\eta_0$ to be $\min_i \frac{-2a_i}{a_i^2+b_i^2}>0.$
\end{proof}
Using \Cref{cl:help}, we conclude that as long as the matrix $\mat{M}$ below has eigenvalues with real part negative then $(\vec{x}^*,\vec{y}^*)$ is attracting:

\begin{equation}\label{eq:noJ}
\mat{M} \defeq \left(\begin{array}{cc}
 k\cdot \mat{I} +  \left(\begin{array}{cc} -\nabla^2_{\vec{x}\vec{x}}U(\vec{x}^*,\vec{y}^*) &-\nabla^2_{\vec{y}\vec{x}}U(\vec{x}^*,\vec{y}^*)\\\nabla^2_{\vec{x}\vec{y}}U(\vec{x}^*,\vec{y}^*) &\nabla^2_{\vec{y}\vec{y}}U(\vec{x}^*,\vec{y}^*)  \end{array}\right) & -  k\cdot \mat{I}
\\  p\cdot \mat{I} &  - p\cdot \mat{I}
\end{array}\right).
\end{equation}
By setting $\mat{M}_{11},\mat{M}_{12},\mat{M}_{21},\mat{M}_{22}$ as the corresponding block matrices we compute the characteristic polynomial of $\mat{M}$. It holds
    \begin{gather}
        \det ( \mat{M} - \lambda \mat{I} ) = \det \left(  \begin{bNiceArray}{c|c}[margin]
                \mat{M}_{11} & \mat{M}_{12} \\ \hline
                \mat{M}_{21} & \mat{M}_{22} 
            \end{bNiceArray} - \lambda \mat{I} \right) = 
        \det \left ( 
        \begin{bNiceArray}{c|c}[margin]
                \mat{M}_{11} - \lambda \mat{I} & -  k\mat{I} \\ \hline
                p\mat{I} &  - p\mat{I} - \lambda \mat{I} \\
            \end{bNiceArray}
        \right) = \\
        = 
        \det \left ( 
        \begin{bNiceArray}{c|c}[margin]
                \mat{M}_{11} - \lambda \mat{I}  - \frac{kp}{p+\lambda}\mat{I}  & -  k\mat{I} \\ \hline
                0 &  -p \mat{I} - \lambda \mat{I} \\
            \end{bNiceArray}
        \right) = \\
        =
        \det \left( \mat{M}_{11} - \left(\lambda+\frac{kp}{p+\lambda}\right)\mat{I} 
        \right) \det \left( - (p + \lambda)\mat{I} \right).
    \end{gather}
     $\det ( \mat{M} - \lambda \mat{I} )$ has roots at $\lambda = -p$ and when $$\det \left( \mat{M}_{11} - \left(\lambda+\frac{kp}{p+\lambda}\right)\mat{I} 
        \right)=0.$$
Therefore, if $\mat{M}$ has all eigenvalues with real part negative, we must have $p>0.$
        Let $\rho$ be an eigenvalue of $\left(\begin{array}{cc} -\nabla^2_{\vec{x}\vec{x}}U(\vec{x}^*,\vec{y}^*) &-\nabla^2_{\vec{y}\vec{x}}U(\vec{x}^*,\vec{y}^*)\\\nabla^2_{\vec{x}\vec{y}}U(\vec{x}^*,\vec{y}^*) &\nabla^2_{\vec{y}\vec{y}}U(\vec{x}^*,\vec{y}^*)  \end{array}\right).$ It must hold that \begin{equation}\label{eq:invest}
            \lambda + \frac{kp}{p+\lambda} -k = \rho.
        \end{equation}
        We need to investigate under what assumptions $\lambda$ will have a real part negative. We expand \Cref{eq:invest} and we get
        \begin{equation}\label{eq:quadr}
        \lambda^2 + \lambda(p-k - \rho) - \rho p =0. 
        \end{equation}
We solve Equation \Cref{eq:quadr} and we get
        \begin{equation}\label{eq:lambda}
        \lambda_{1,2} = \frac{-p+k+\rho\pm \sqrt{(p-k-\rho)^2 + 4\rho p}}{2}.
        \end{equation}
        We provide sufficient conditions so that $\textrm{Re}(\lambda_{1,2}) <0.$
We consider the following cases:
\begin{itemize}
\item $\rho$ is real and negative. In this case observe that for $p=0$ we have that $\lambda_{1,2} = k+\rho \textrm{ and }0.$ Note that for $k<0$, the first eigenvalue has real part negative and the second eigenvalue is zero.
We compute the first derivative at $p=0$ and this gives $\frac{1}{2}\left(-1+\frac{-k+\rho}{|k+\rho|}\right) = \frac{\rho}{-k-\rho}<0.$ Therefore the $\lambda_{1,2}$ as a function of $p$ is strictly decreasing at zero. We conclude that for $p$ sufficiently small and positive, both eigenvalues of $\mat{M}$ will be real and negative. 
Hence if an eigenvalue of $\mat{H}$ is real and negative then for $p$ sufficiently small positive and $k<0$ both eigenvalues of $\mat{M}$ will be negative.
\item $\rho$ is complex with $|\textrm{Im}(\rho)|\neq 0.$ In this case observe that for $p=0$ we have that $\lambda_{1,2} = k+\rho \textrm{ and }0.$ If we choose $\min(0, -\textrm{Re}(\rho))>k$, then the first eigenvalue will have real part negative. Now using the same idea as for the first case, we compute the first derivative of the real part of $\lambda_{1,2}$ as a function of $p$ and we get
\[
\frac{1}{2}  \left(-1\pm \frac{k^2 - \textrm{Re}(\rho)^2 - \textrm{Im}(\rho)^2}{(k+\textrm{Re}(\rho))^2+\textrm{Im}(\rho)^2}\right).
\]
The expression above is negative as long as $-\textrm{Re}(\rho)^2-\textrm{Im}(\rho)^2 - k\textrm{Re}(\rho)<0.$ The equation is trivially satisfied when $\textrm{Re}(\rho)<0$ since $k$ is chosen to be negative. Assume $\textrm{Re}(\rho)> 0$ (if it is zero then the above is trivially true since the imaginary part is non-zero). It occurs that the inequality above is true when $k > -\frac{|\rho|^2}{\textrm{Re}(\rho)}.$ 
Since $k$ is chosen to be smaller than $\min(0, -\textrm{Re}(\rho)),$ it suffices to show that $$-\textrm{Re}(\rho)) > -\frac{|\rho|^2}{\textrm{Re}(\rho)} \textrm{ is satisfied by our assumptions.}$$
The above is true as long as $\textrm{Im}(\rho)\neq 0.$ Let $E$ be the set of eigenvalues $\rho$ of $\mat{H}$ with real part positive (and non-zero imaginary part by assumption). There exists a choice for $\eta, k,p$, in which $\eta,p>0$ and sufficiently small and $k$ is negative and chosen to be $$\min_{\rho\in E} \frac{|\rho|^2}{\textrm{Re}(\rho)}>-k>\max_{\rho\in E} \textrm{Re}(\rho).$$
\end{itemize}
}

\newpage

\section{Experiments}

We elaborate on the experiment mentioned in the main body.
Our experiment includes a dataset of a mixture of 2-D Gaussians with 8 modes. Our architecture includes 8 ``shallow'' generators and discriminators with 2 layers of 2-16-2 ReLU activations, compared with a ``large'' single-agent GAN with 4 layers of 2-128-256-1024-2 activations. Interestingly, the giant one fails in a double sense; It demonstrates both mode-collapsing and mode-drop phenomena without stabilizing. On the other hand{,} our architecture with a small number of neurons achieves to {fit the data well.}

The iterations illustrate the step at which the model roughly stops learning/converges. We observed no change after the demonstrated number of iterations.

\begin{center}
\scalebox{0.85}{
    \addtolength{\leftskip} {-5cm}
    \addtolength{\rightskip}{-5cm}
\begin{tabular}{|c|c|c|c|}\hline
    \textrm{Configuration} & \textrm{Generator(s)} & \textrm{Discriminator(s)} & \textrm{Typical outcome} \\
    \hline
     \text{Single GAN} &  \# $1 \times$  \Big\{\begin{tabular}{l} {\footnotesize 2-128-256-1024-2}  \\ {\footnotesize \textrm{Linear w/} ~\text{ReLU}} \end{tabular} &  \# $1~ \times$  \Big\{\begin{tabular}{l} {\footnotesize 2-128-256-1024-2}  \\ {\footnotesize \textrm{Linear w/} ~\text{ReLU}} \end{tabular} &  
        \begin{tabular}{r} { \small Mode coll. }  {\small in $\sim 1\times10^3$~ \text{iters.} }  \\ { \small Mode drop }  {\small in $\sim 3\times10^3$~ \text{iters.} } \end{tabular} \\
    \hline
     \text{Multi-GAN}  & \# $8  \times$  \Big\{\begin{tabular}{l} {\footnotesize 2-16-2 }  \\ {\footnotesize \textrm{Linear w/} ~\text{ReLU}} \end{tabular} &  \# $8 ~ \times$ \Big\{\begin{tabular}{l} {\footnotesize 2-16-2 }  \\ {\footnotesize \textrm{Linear w/} ~\text{ReLU}} \end{tabular}  & \begin{tabular}{l} { \small Distr. learned }\\ {\small in $\sim 3\times10^3$~ \text{iters.} } \end{tabular} \\ \hline
\end{tabular}
}
\end{center}


\begin{figure}[h!]
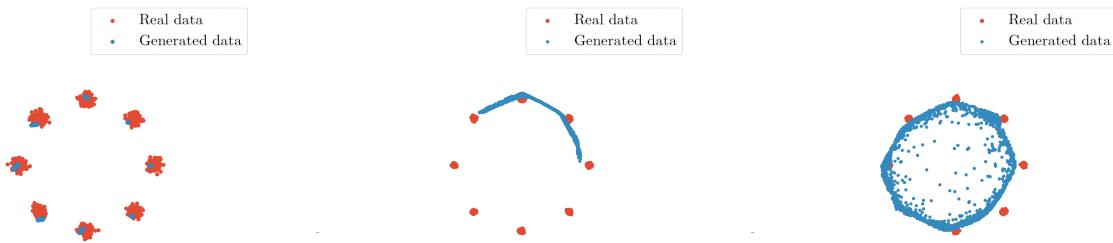

\centering
\hspace{-60pt}
\begin{subfigure}[b]{0.35\textwidth}
\includegraphics[width=1\textwidth]{figures/real-vs-generated-crop.png}%
\end{subfigure}%
\begin{subfigure}[b]{0.35\textwidth}
\includegraphics[width=1\textwidth]{figures/mode-collapse-crop.png}
\end{subfigure}%
\begin{subfigure}[b]{0.35\textwidth}
\includegraphics[width=1\textwidth]{figures/big-generated-crop.png}%
\end{subfigure}
\vspace{-30pt}
\caption{From left to right: $(i)$ Each generator of MGAN learns one mode of a mixture of 8 gaussians, $(ii)$ Mode Collapse of single-agent GAN's, $(iii)$ Single-agent GAN can't discriminate between the modes.}
\end{figure}
\vspace{-2em}

\subsection{MGAN vs. WGAN}
\label{sec:experiments_gans}
CIFAR-10 is a {well-established} {testbed} for various GAN architectures. It contains 10 balanced classes of images of different objects.

Just to further illustrate the merits of multi-agent GAN's, we offer a selection of images generated by WGAN-GP \citep{arjovsky2017wasserstein} and the MGAN \citep{hoang2017multi}.

In the pictures provided by the WGAN-GP we see that across iterations, the generated samples tend to cover only certain classes of the dataset while it takes longer for the samples to become realistic.

On the other hand, the MGAN architecture from an early stage provides with diverse samples that tend to be more realistic from early on.
{
\begin{figure}[H]
    \centering
        {\includegraphics{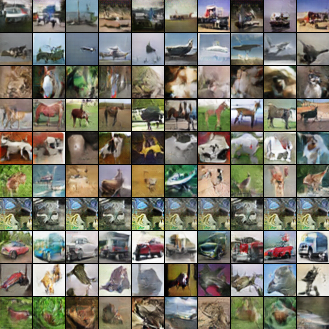}}
        \caption{MGAN after 300 iterations}
\end{figure}  
\begin{figure}[H]
    \centering
    \includegraphics{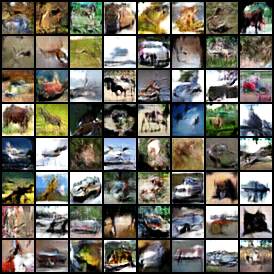}
    \caption{WGAN after 10000 iterations}
     \includegraphics{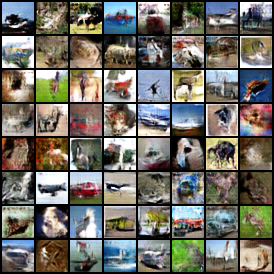}
    \caption{WGAN after 15000 iterations}
\end{figure}
\begin{figure}[H]
    \centering
     \includegraphics{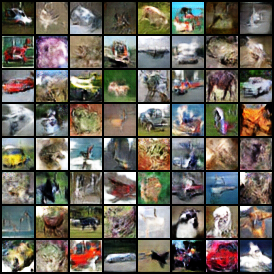}
    \caption{WGAN after 25000 iterations}
\end{figure}
}

\pagebreak

{
\section{
Empirical performance of the proposed method}
In this section, we perform a number of empirical experiments to assess the performance of our proposed method. As we highlighted, the conditions of \Cref{thm:sufficient} are only sufficient and not necessary. In the following figure, we plot the curves of the \textit{Nash equilibrium gap} for $500$ random two-team zero-sum games. In every such game, every player has $2$ available strategies. We note that for all of these games we use a random initialization for the players' strategies and the same parameters $\eta = 0.05, ~k=-1.2, ~p=0.02$ for our proposed method. It is notable that the Nash equilibrium gap vanishes for virtually all of these instances. This indicates that the conditions of \Cref{thm:sufficient} are only sufficient and a more complex analysis could show less restrictive conditions for convergence with the same simple parametrization.

For completeness, we include the definition of the Nash equilibrium gap:
\begin{definition}[NE-gap]
    We define the Nash equilibrium gap as the sum of the differences between the best-response and the strategy every player plays:
    \begin{equation}
        \mathrm{NE\text{-}gap}(\vec{x}) = \sum_{i=1}^n \max_{\vec{x}_i^*}\{ u_i( \vec{x}_i^*, \vec{x}_{-i}) \} - u_i(\vec{x}).
    \end{equation}
\end{definition}

\begin{figure}[h!]
    \centering
     \includegraphics[width=1\textwidth]{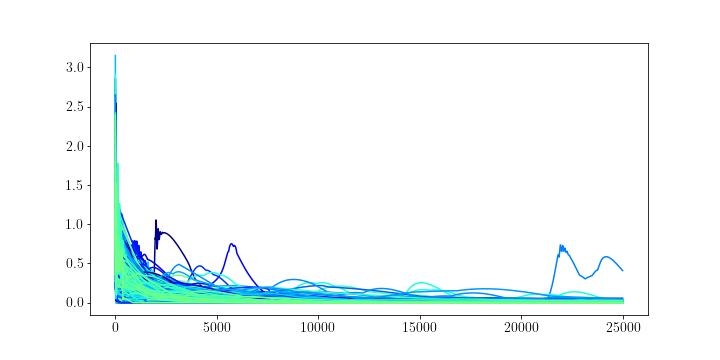}
    \caption{Nash equilibrium gap for $2$-vs.-$2$ two-team zero-sum game with $2$ strategies per player. KPV with parameters $\eta = 0.05, ~k=-1.2, ~p=0.02$.}
\end{figure}

Additionally, we apply the same method with fixed parameters $\eta = 0.05, ~k=-1.05, ~p=0.005$ to two-team zero-sum games with $3$ strategies per player and $2$ players in each team. Again, most random games seem to converge to a $0$ NE-gap.
\begin{figure}[h!]
    \centering
     \includegraphics[width=1\textwidth]{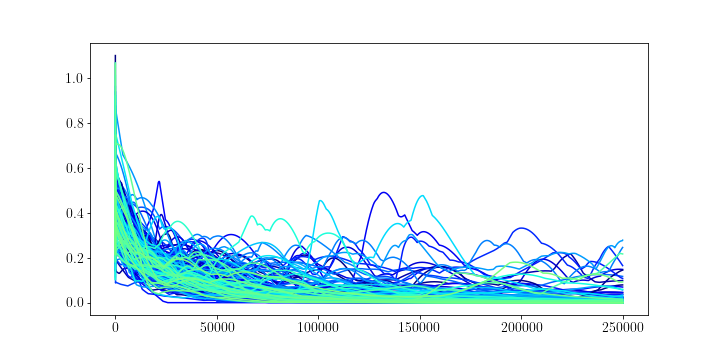}
    \caption{
    Nash equilibrium gap for $2$-vs.-$2$ two-team zero-sum game with $3$ strategies per player. KPV with parameters $\eta = 0.05, ~k=-1.05, ~p=0.005$.}
\end{figure}

Lastly, we use the same method with fixed parameters $\eta = 0.1, ~k=-1.15, ~p=0.01$ for $500$ random two-player zero-sum games where each player has $5$. The method converges for virtually all games despite of not tuning the parameters for every game.
\begin{figure}[h!]
    \centering
     \includegraphics[width=1\textwidth]{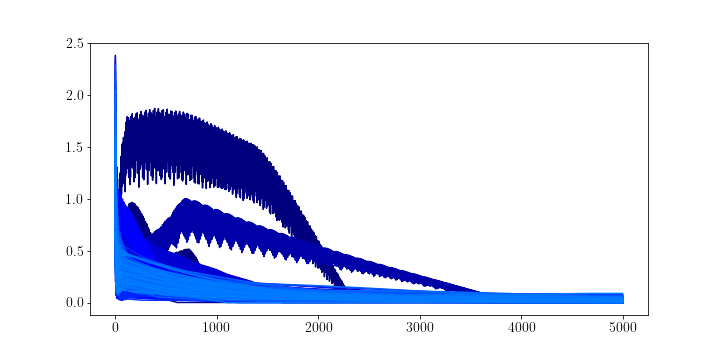}
    \caption{
    Nash equilibrium gap for $2$-vs.-$2$ two-team zero-sum game with $3$ strategies per player. KPV with parameters $\eta = 0.05, ~k=-1.05, ~p=0.005$.}
\end{figure}
}

\end{document}